\RequirePackage{amsmath}
\documentclass[english,twocolumn]{article}
\usepackage[utf8]{inputenc}
\usepackage[big]{dgruyter}
\setlocalecaption{ngerman}{abstract}{Zusammenfassung}

\usepackage{abbreviations}
\usepackage{tikz,pgfplots}
\pgfplotsset{compat=newest}
\usepackage{aligned-overset}
\usepackage{siunitx}
\usepackage{multirow}
\usepackage{nicefrac}
\usepackage{cite}

\usepackage{wrapfig}

\theoremstyle{dgthm}
\newtheorem{theorem}{Theorem}

\newtheorem{proposition}[theorem]{Proposition}
\newtheorem{corollary}[theorem]{Corollary}
\newtheorem{remark}[theorem]{Remark}
\theoremstyle{dgdef}
\newtheorem{definition}[theorem]{Definition}
\newtheorem{assumption}[theorem]{Assumption}

\newcommand{\SOS}{\mathrm{SOS}}

\begin{document}
  \author*[1]{Robin Strässer}
  \author[2]{Julian Berberich}
  \author[3]{Manuel Schaller}
  \author[4]{Karl Worthmann}
  \author[2]{Frank Allgöwer} 
  \runningauthor{R.\ Strässer et al.}
  \affil[1]{University of Stuttgart, Institute for Systems Theory and Automatic Control, 70550 Stuttgart, Germany, e-mail: robin.straesser@ist.uni-stuttgart.de}
  \affil[2]{University of Stuttgart, Institute for Systems Theory and Automatic Control, 70550 Stuttgart, Germany}
  \affil[3]{Chemnitz University of Technology, Faculty of Mathematics, 09111 Chemnitz, Germany}
  \affil[4]{Technische Universität Ilmenau, Institute of Mathematics, 98693 Ilmenau, Germany}
  \title{Koopman-based control of nonlinear systems with closed-loop guarantees}
  \runningtitle{Koopman-based control of nonlinear systems}
  \abstract{
  In this paper, we provide a tutorial overview and an extension of a recently developed framework for data-driven control of unknown nonlinear systems with rigorous closed-loop guarantees.
  The proposed approach relies on the Koopman operator representation of the nonlinear system, for which a bilinear surrogate model is estimated based on data.
  In contrast to existing Koopman-based estimation procedures, we state guaranteed bounds on the approximation error using the stability- and certificate-oriented extended dynamic mode decomposition (SafEDMD) framework.
  The resulting surrogate model and the uncertainty bounds allow us to design controllers via robust control theory and sum-of-squares optimization, guaranteeing desirable properties for the closed-loop system.
  We present results on stabilization both in discrete and continuous time, and we derive a method for controller design with performance objectives.
  The benefits of the presented framework over established approaches are demonstrated with a numerical example.  
  }
  \startpage{1}

\maketitle

\section{Introduction}
Data-driven controller design is a topic of significant interest in recent years.
When controlling complex and safety-critical systems based on data, guaranteeing desired properties such as stability and robustness for the closed loop is paramount. 
For linear systems, the literature contains a variety of approaches that admit rigorous theoretical guarantees and which were successfully used in challenging control applications; see~\cite{hou:wang:2013,markovsky:dorfler:2021,waarde2023informativity, faulwasser:ou:pan:schmitz:worthmann:2023,berberich2024overview} for overview articles.
Beyond the linear case, however, theoretical guarantees are less explored and require a tailored treatment depending on specific system classes~\cite{martin:schon:allgower:2023b,berberich2024overview}.

The Koopman operator proposed in~\cite{koopman:1931} provides a powerful foundation for rigorous and practical nonlinear data-driven analysis, prediction, and control, see~\cite{mezic:2005,proctor:brunton:kutz:2018,otto:rowley:2021,mauroy:mezic:susuki:2020,bevanda:sosnowski:hirche:2021} for recent works in this direction.
The basic idea is to represent a nonlinear dynamical system as a potentially infinite-dimensional linear system in suitably chosen lifted coordinates.
For nonlinear systems with input, one obtains a lifted linear~\cite{proctor:brunton:kutz:2018,korda:mezic:2018a} or bilinear~\cite{peitz:otto:rowley:2020,williams:hemati:dawson:kevrekidis:rowley:2016} system which can be used to design a controller for the underlying nonlinear system~\cite{surana:2016,huang:ma:vaidya:2018,bramburger2024synthesizing}.
A key strength of the Koopman operator is that a representation can be estimated based on data using extended dynamic mode decomposition (EDMD)~\cite{williams:kevrekidis:rowley:2015,korda:mezic:2018b} as an efficient regression-based tool.
As a consequence, Koopman-based control approaches were successfully leveraged in various applications~\cite{bruder:fu:gillespie:remy:vasudevan:2021,kim:quan:chung:2023,budisic:mohr:mezic:2012}.
In order to obtain rigorous guarantees for the controlled system, it is important to quantify the approximation error and take it into account in the closed-loop analysis.
Probabilistic error bounds were derived by~\cite{chen:vaidya:2019,nuske:peitz:philipp:schaller:worthmann:2023, mezic:2022,zhang:zuazua:2023}. 
Refining these towards proportional error bounds in the sense that the upper bound depends on the distance to the desired set point enables data-driven model predictive control schemes with guaranteed practical stability~\cite{bold:grune:schaller:worthmann:2025} and stabilizing controllers based on robust control theory~\cite{strasser:berberich:allgower:2023a}.

In this paper, we provide a tutorial overview and extension of a recently developed framework for data-driven controller design for unknown nonlinear systems with rigorous closed-loop guarantees.
The proposed approach relies on the stability- and certificate-oriented EDMD (SafEDMD) approach~\cite{strasser:schaller:worthmann:berberich:allgower:2025,strasser:schaller:worthmann:berberich:allgower:2024b} which consists of two main steps:
1) estimating a bilinear surrogate model of the Koopman operator from data while providing rigorous proportional bounds on its approximation error;
2) designing a controller which robustly stabilizes the uncertain surrogate model and, therefore, is guaranteed to stabilize the unknown nonlinear system.
We will introduce this framework both in continuous time~\cite{strasser:schaller:worthmann:berberich:allgower:2025}, assuming that state derivative measurements are available, and in discrete time~\cite{strasser:schaller:worthmann:berberich:allgower:2024b}, where state and input measurements suffice to stabilize an unknown system via sampled-data feedback.
SafEDMD provides the first Koopman-based control approach for nonlinear systems with rigorous end-to-end guarantees, and it can also be used in the presence of noise~\cite{chatzikiriakos:strasser:allgower:iannelli:2024} or to design data-driven model predictive control schemes with terminal conditions~\cite{worthmann:strasser:schaller:berberich:allgower:2024}.
While the original controller design for the bilinear surrogate model from~\cite{strasser:schaller:worthmann:berberich:allgower:2025,strasser:schaller:worthmann:berberich:allgower:2024b} relies on a conservative over-approximation of the bilinear term, more recent results using sum-of-squares (SOS) optimization were developed in the discrete-time treatment, significantly reducing conservatism and improving the practicality of the approach~\cite{strasser:berberich:allgower:2025}.
As a first major contribution of this work and leveraging the SafEDMD framework, we extend this SOS-based approach from discrete-time results obtained in~\cite{strasser:berberich:allgower:2025} to continuous-time systems, deriving a stabilizing SOS-based controller with significantly reduced conservatism compared to the linear matrix inequality (LMI)-based approach from~\cite{strasser:schaller:worthmann:berberich:allgower:2025}. 
Secondly, as existing Koopman-based control approaches with rigorous guarantees mainly focus on closed-loop stability, we provide an extension towards performance-oriented criteria. 
Based on the SafEDMD framework, we derive Koopman-based controllers with rigorous guarantees on a desired performance specification, e.g., a bounded $\mathcal{L}_2$-gain, for the closed-loop system.
Notably, these guarantees are provided for the bilinear surrogate model and the proportional error bound obtained via SafEDMD and, thereby, for the unknown nonlinear system.
We address this problem in the continuous-time domain and leave the adaptation to discrete time for future research.

The paper is structured as follows.
In Section~\ref{sec:problem-setting}, we state the problem setup.
Next, Section~\ref{sec:bilinear-surrogate} presents the SafEDMD learning architecture and the guaranteed error bounds.
We then use the resulting bilinear surrogate model for designing controllers which guarantee stability (Section~\ref{sec:controller-design}) and performance (Section~\ref{sec:controller-design-performance}) for the closed-loop system.
We present numerical results in Section~\ref{sec:numerics} and conclude the paper in Section~\ref{sec:conclusion}.
\\

\noindent \textbf{Notation}:
An ordered set of integers $[a,b]\cap \bbZ$ is denoted by $[a:b]$.
We write $I_p$ for the $p\times p$ identity matrix and $0_{p\times q}$ for the $p\times q$ zero matrix, where the subscript indices are omitted if the dimension is clear from the context. 
For symmetric $A$, we write $A\succ 0$ ($A\succeq 0$) if $A$ is positive (semi)definite and define negative (semi)definiteness analogously.
We abbreviate $A^\top B A$ by writing $[\star]^\top B A$ and denote matrix blocks which can be inferred from symmetry by $\star$.
The Kronecker product is denoted by $\kron$.
We use a vectorial multi-index $\alpha\in\bbN_0^n$ to define monomials $x^\alpha = x_1^{\alpha_1}\cdots x_n^{\alpha_n}$, where $|\alpha|=\alpha_1+\cdots + \alpha_n$.
The set of all polynomials $s$ with $s(x) = \sum_{\alpha\in\bbN_0^n,|\alpha|\leq d} s_\alpha x^\alpha$, where $s_\alpha\in\bbR$ and $|\alpha|\leq d$ for some degree $d\in\bbN_0$ is denoted by $\bbR[x,d]$. 
Similarly, we define the set of all polynomial matrices of size $p\times q$ with elements in $\bbR[x,d]$ by $\bbR[x,d]^{p\times q}$.
We call a matrix $S\in\bbR[x,2d]^{p\times p}$ an SOS matrix, denoted by $S\in\SOS[x,2d]^p$, if it can be decomposed as $S=T^\top T$ for some $T\in\bbR[x,d]^{q\times p}$ with~$q \in \mathbb{N}$. 
A matrix $S\in\SOS[x,2d]^{p}$ is said to be strictly SOS if $(S-\varepsilon I)\in\SOS[x,2d]^p$ for some $\varepsilon>0$, which we denote by $S\in\SOS_+[x,2d]^p$. 
\section{Problem setting}\label{sec:problem-setting}
We consider control-affine nonlinear systems
\begin{equation}\label{eq:dynamics-nonlinear}
    \dot{x}(t) 
    = f_\mathrm{c}(x(t),u(t))
    = f(x(t)) + \sum_{i=1}^m g_i(x(t))u_i(t)
\end{equation}
with state $x(t)\in\bbR^n$, input $u(t)\in\bbR^m$ at time $t>0$, $x(0) = \hat{x}$, and unknown drift $f:\bbR^n\to\bbR^n$ and input maps $g_i:\bbR^n\to\bbR^n$.
Further, we assume that the origin is a controlled equilibrium, i.e., $f(0)=0$.
Throughout the paper, we provide both a direct treatment of~\eqref{eq:dynamics-nonlinear} in continuous time as well as a sampled-data treatment in discrete time.
For the latter, we sample the system via zero-order hold, i.e., we sample equidistantly in time for piecewise constant input $u(t)=u_k$ for $t\in[t_k,t_{k+1})$ for some sampling period $\Delta t >0$ and $t_k=k\Delta t$, $k\in\bbN_0$.
This leads to the discrete-time system dynamics 
\begin{equation}\label{eq:dynamics-nonlinear-sampled}
    x_{k+1}
    = f_\mathrm{d}(x_k,u_k)
    = x_k + \!\!\int\limits_{t_k}^{t_{k+1}}\!\! f(x(t)) + g(x(t)) u_k \dd t.
\end{equation}
Since the system dynamics are assumed to be unknown, we use Koopman theory and EDMD to derive a data-driven bilinear surrogate model of the nonlinear systems~\eqref{eq:dynamics-nonlinear} and~\eqref{eq:dynamics-nonlinear-sampled} based on collected data samples. 
The obtained surrogate model then allows us to design a controller with closed-loop guarantees for the \emph{unknown} nonlinear system.
In particular, our goal is to design a state-feedback control law such that the origin is exponentially stable (Section~\ref{sec:controller-design}) and satisfies a desired performance specification (Section~\ref{sec:controller-design-performance}).
\section{Koopman-based bilinear surrogate model}\label{sec:bilinear-surrogate}

In this section, we introduce the SafEDMD learning architecture including the estimated surrogate models in both continuous and discrete time, as proposed in~\cite{strasser:schaller:worthmann:berberich:allgower:2025,strasser:schaller:worthmann:berberich:allgower:2024b}.
The basic idea relies on Koopman theory to lift the nonlinear system to an infinite-dimensional bilinear system and derive a finite-dimensional bilinear approximant while explicitly accounting for the remaining residual error.
Before deriving the bilinear surrogate model in continuous time using state-derivative data (Section~\ref{sec:bilinear-surrogate-continuous}) and in discrete time using derivative-free data (Section~\ref{sec:bilinear-surrogate-discrete}), we introduce basic elements of Koopman theory.

In the Koopman approach, introduced almost 100 years ago~\cite{koopman:1931}, nonlinear systems are considered through the lens of observable functions.
Correspondingly, the system is lifted to an infinite-dimensional function space, where its evolution is linear, making it amenable for techniques based on linear regression.
The Koopman operator $\cK_t^u$ for \eqref{eq:dynamics-nonlinear} with constant input $u(t)\equiv u$ acting on an observable $\varphi:\mathbb{X} \to \mathbb{R}$ is defined by 
\begin{equation}
    (\cK_t^u \varphi)(\hat{x}) = \varphi(x(t;\hat{x},u))
\end{equation}
for all times $t\geq 0$ and initial value $\hat{x}\in\bbX$ for some compact set $\bbX\subseteq\bbR^n$. 
Here, $x(t;\hat{x},u)$ denotes the solution of~\eqref{eq:dynamics-nonlinear} at time $t$ for $u(t)\equiv u$ and $x(0)=\hat{x}$.
As a consequence of the semigroup and continuity properties of the flow, it can be shown that $(\cK_t^u)_{t\geq 0}$ on $L^2(\mathbb{X})$ is a strongly continuous semigroup; see, e.g.,~\cite{philipp:schaller:worthmann:peitz:nuske:2024}. 
Thus, we may define its infinitesimal generator $\cL^u$ by
\begin{equation}
    (\cL^u \varphi)(\hat{x}) = \lim_{t\searrow 0} \frac{(\cK_t^u \varphi)(\hat{x}) - \varphi(\hat{x})}{t} = \frac{\dd}{\dd t} \varphi(x(t;\hat{x},u))
\end{equation}
for all $\varphi$ in the domain $D(\cL^u)$, which consists of all $L^2(\mathbb{X})$-functions for which the limit exists.
Notably, the Koopman generator preserves control affinity~\cite{williams:hemati:dawson:kevrekidis:rowley:2016,surana:2016}, i.e., 
\begin{equation}\label{eq:generator-control-affine}
    \cL^u = \cL^0 + \sum_{i=1}^m u_i (\cL^{e_i} - \cL^0)
\end{equation}
holds for any $u\in\bbR^m$ with the Koopman generators $\cL^0$, $\cL^{e_i}$, $i\in[1:m]$, corresponding to the constant control functions $u(t) \equiv 0$, $u(t) \equiv e_i$ for the unit vector $e_i$, $i\in[1:m]$, respectively. 
In contrast to the Koopman generator $\cL^u$, the Koopman operator $\cK_t^u$ only approximately inherits this structure~\cite{peitz:otto:rowley:2020}, where the resulting approximation error is analyzed in~\cite[Thm. 3.1]{strasser:schaller:worthmann:berberich:allgower:2024b}.

Both the Koopman operator as well as the Koopman generator are \emph{linear} yet infinite-dimensional operators and, hence, practical implementations rely on finite-dimensional approximations~\cite{mauroy:mezic:susuki:2020,bevanda:sosnowski:hirche:2021,budisic:mohr:mezic:2012}.
In this work, we leverage EDMD~\cite{proctor:brunton:kutz:2018,williams:kevrekidis:rowley:2015,williams:hemati:dawson:kevrekidis:rowley:2016} to learn the Koopman action on a pre-defined $(N+1)$-dimensional subspace $\mathbb{V}$ spanned by the components of the vector-valued observable function $\Phi(x)=\begin{bmatrix} 1 & \widehat{\Phi}(x)^\top \end{bmatrix}^\top$, where $\widehat{\Phi}$ satisfies
\begin{equation}\label{eq:observables-Lipschitz}
    \widehat{\Phi}(0)=0 
    \quad\text{and}\quad
    \|x\| \leq \|\widehat{\Phi}(x)\| \leq L_\Phi \|x\|
\end{equation}
for all $x \in \mathbb{X}$ and some $L_\Phi > 0$.
In particular, we leverage the SafEDMD framework~\cite{strasser:schaller:worthmann:berberich:allgower:2025,strasser:schaller:worthmann:berberich:allgower:2024b} providing an EDMD-based surrogate model with rigorous error certificates.
The constant observable $\phi_0(x)\equiv 1$ is required to enable the representation of linear systems via the surrogate model, see also the discussion at the end of Section~\ref{sec:bilinear-surrogate-continuous}.
The conditions in~\eqref{eq:observables-Lipschitz} are satisfied, e.g., for the choice
\begin{equation}
    \widehat{\Phi}(x) = \begin{bmatrix}
        x^\top & \phi_{n+1}(x) & \cdots & \phi_N(x)
    \end{bmatrix}^\top,
\end{equation}
where $\phi_\ell$, $\ell\in[n+1,N]$ are continuously differentiable in $\bbX$ and satisfy $\phi_\ell(0)=0$.

\subsection{Continuous-time surrogate model}\label{sec:bilinear-surrogate-continuous}

The approach from~\cite{strasser:schaller:worthmann:berberich:allgower:2025} exploits the structure of the chosen observables in $\Phi(x)$. 
In particular, the constant observable $\phi_0(x)\equiv 1$ implies that the row of any finite-dimensional matrix approximation of $\cL^u$ corresponding to $\phi_0$ should be zero to preserve $(\cL^u \phi_0)(x(t)) = \ddt{}\phi_0(x(t;\hat{x},u)) = 0$. 
Moreover, we exploit $f(0)=0$ to compute 
\begin{equation}
    (\cL^0 \Phi)(0) = \nabla \Phi(0)^\top f(0) = 0,
\end{equation}
where we understand the application of operators to vector-valued functions in a component-wise fashion.
Computing the respective finite-dimensional approximation $L_d^0\Phi(0)=\begin{bmatrix}
    0 \\ L_{d,21}^0
\end{bmatrix}$ yields that
the column $L_{d,21}^0$ corresponding to $\phi_0$ should also be chosen as the zero vector.
Thus, to respect these properties of the Koopman generators $\cL^0$, $\cL^{e_i}$ in finite-dimensional data-driven approximations, we define
\begin{equation}\label{eq:structure-generator}
    L_d^0 = \begin{bmatrix}
        0 & 0_{1\times N} \\
        0_{N\times 1} & A
    \end{bmatrix},
    \; 
    L_d^{e_i} = \begin{bmatrix}
        0 & 0_{1\times N} \\ 
        b_{i} & B_i 
    \end{bmatrix}
\end{equation}
with suitable matrices $A,B_i\in \bbR^{N\times N}$, and vectors $b_i\in \bbR^N$, $i\in[1:m]$.
These are constructed using data samples
\begin{equation}\label{eq:data-continuous}
    \cD=\left\{\{x_j^{\bar{u}}, \dot{x}_j^{\bar{u}}\}_{j=1}^d \text{ with } \bar{u}\in\{0,e_1,...,e_m\} \right\}   
\end{equation}
for $(m+1)$ repeated experiments of autonomous variants of the controlled system with the constant control input $u(t) \equiv \bar{u}$.
Note that we choose the canonical unit vectors as control values for the data generation. 
While a generalization to an arbitrary basis of $\bbR^m$ is straightforward, it still requires collecting $m+1$ samples at each (state) data point; see~\cite{bevanda:driessen:iacob:toth:sosnowski:hirche:2024,bold:philipp:schaller:worthmann:2024} for more flexible sampling schemes w.r.t the control input $u$.
Based on $\cD$, the elements $A$, $b_i$, and $B_i$, $i\in [1:m]$, are chosen as the solution to the linear regression problems
\begin{subequations}\label{eq:linear regression}
    \begin{align}
        A &= \argmin_{A\in\bbR^{N\times N}} \left\| Y^0 - A X^0\right\|_\mathrm{F},
        \\
        \begin{bmatrix}
            b_i & B_i 
        \end{bmatrix}
        &= \argmin_{\substack{b_i\in\bbR^N\\B_i\in\bbR^{N\times N}}} \left\| Y^{e_i} - \begin{bmatrix}
            b_i & B_i
        \end{bmatrix} X^{e_i}\right\|_\mathrm{F}
    \end{align}
\end{subequations}
with the lifted data matrices 
\begin{subequations}
    \begin{align}
        X^{0} &= \begin{bmatrix}
            \widehat{\Phi}(x_1^0) & \cdots & \widehat{\Phi}(x_d^0)
        \end{bmatrix},
        \label{eq:lifted-data-X0}
        \\
        X^{e_i} &= \begin{bmatrix}
            \Phi(x_1^{e_i}) & \cdots & \Phi(x_d^{e_i})
        \end{bmatrix},
        \label{eq:lifted-data-Xei}
        \\
        Y^{\bar{u}} &= \begin{bmatrix}
            \nabla \widehat{\Phi}(x_1^{\bar{u}})^\top \dot{x}_1^{\bar{u}}
            & \cdots & 
            \nabla \widehat{\Phi}(x_d^{\bar{u}})^\top \dot{x}_d^{\bar{u}}
        \end{bmatrix}
        \label{eq:lifted-data-Yu-continuous}
    \end{align}
\end{subequations}
for $\bar{u}\in\{0,e_1,...,e_m\}$.
By exploiting the control affinity in~\eqref{eq:generator-control-affine}, the Koopman-based dynamics $\ddt{}\Phi(x) = \mathcal{L}^u\Phi(x)$ on the lifted space may be approximated by the bilinear dynamics based on the estimators $L_d^0$ and $L_d^{e_i}$, $i\in [1:m]$, defined above, that is,
\begin{equation}\label{eq:bilinear-surrogate-continuous}
    \ddt{}\widehat{\Phi}(x) \approx A \widehat{\Phi}(x) +  B_0 u + \sum_{i=1}^m u_i(B_i - A) \widehat{\Phi}(x)
\end{equation}
with $B_0=\begin{bmatrix}
    b_1 & \cdots & b_m
\end{bmatrix}$.
Note that the constant observable $\phi_0(x)\equiv 1$ is only implicitly included in the bilinear dynamics via the linear control term $B_0 u$, which allows considering the reduced observable function $\widehat{\Phi}(x)$ satisfying $\widehat{\Phi}(0)=0$.
Here, the dynamics hold only approximately since we estimate the unknown system matrices via \emph{finite} data samples and restrict the Koopman action onto a \emph{finite} dictionary. 
The resulting error is rigorously analyzed in Section~\ref{sec:error-bounds}.

\subsection{Discrete-time surrogate model}\label{sec:bilinear-surrogate-discrete}
The main drawback of the continuous-time surrogate model in~\eqref{eq:bilinear-surrogate-continuous} is the reliance on state-derivative measurements; cf.~\eqref{eq:data-continuous}. 
As derivative data is typically hard to obtain, we formulate a discrete-time surrogate model~\cite{strasser:schaller:worthmann:berberich:allgower:2024b} in the following, which only needs data samples of the state.
To this end, for a sampling period $\Delta t>0$, we consider the discretized system dynamics~\eqref{eq:dynamics-nonlinear-sampled} and the corresponding Koopman operator $\cK_{\Delta t}^u$ satisfying $\Phi(x_{k+1})=(\cK_{\Delta t}^{u_k}\Phi)(x_k)$.
Due to the structure of the vector-valued observable~$\Phi$ including the constant observable $\phi_0(x) \equiv 1$ and $f(0)=0$, we have $\phi_0(x_+) = \phi_0(x_k) \equiv 1$ and $(\cK_t^0 \Phi)(0)=\Phi(0)$. 
Thus, we may, similarly as for the Koopman generator in Section~\ref{sec:bilinear-surrogate-continuous}, leverage a particular structure of the Koopman operators $\cK_{\Delta t}^0$, $\cK_{\Delta t}^{e_i}$.
Analogous to~\eqref{eq:structure-generator}, we define
\begin{equation}\label{eq:structure-operator}
    K_{\Delta t,d}^0 = \begin{bmatrix}
        1 & 0_{1\times N} \\
        0_{N\times 1} & A
    \end{bmatrix},
    \; 
    K_{\Delta t,d}^{e_i} = \begin{bmatrix}
        1 & 0_{1\times N} \\ 
        b_{i} & B_i 
    \end{bmatrix}
\end{equation}
for $i\in[1:m]$.
To construct $A$, $b_i$, $B_i$, $i\in [1:m]$ we collect data 
\begin{equation}\label{eq:data-discrete}
    \cD=\left\{
        \{x_j^{\bar{u}},x(\Delta t;x_j^{\bar{u}},\bar{u})\}_{j=1}^d \text{ with } \bar{u}\in\{0,e_1,...,e_m\} 
    \right\}   
\end{equation}
in $(m+1)$ repeated experiments of autonomous variants of the controlled system with the constant control input $u(t) \equiv \bar{u}$.
Here, we emphasize that only pairs of states and successor states are necessary but no state-derivative data.
Based on the collected data, we solve the linear regression problems~\eqref{eq:linear regression} to obtain the unknowns $A$, $b_i$, $B_i$ based on the data matrices $X^0$ in~\eqref{eq:lifted-data-X0}, $X^{e_i}$ in~\eqref{eq:lifted-data-Xei}, and 
\begin{equation}\label{eq:lifted-data-Yu-discrete}
    Y^{\bar{u}} = \begin{bmatrix}
        \widehat{\Phi}(x(\Delta t;x_1^{\bar{u}},\bar{u}))
        & \cdots & 
        \widehat{\Phi}(x(\Delta t;x_d^{\bar{u}},\bar{u}))
    \end{bmatrix}
\end{equation}
for $\bar{u}\in\{0,e_1,...,e_m\}$. 
With slight abuse of notation and to unify the framework, we use the same symbols $A$, $b_i$, $B_i$ to denote matrices of the bilinear surrogate model in continuous and discrete time, although they result from different estimation procedures, i.e., using~\eqref{eq:lifted-data-Yu-continuous} and~\eqref{eq:lifted-data-Yu-discrete} in~\eqref{eq:linear regression}, respectively.  
As for the continuous-time surrogate, we enforce bilinear dynamics with the lifted state $\widehat{\Phi}(x)$, i.e., we approximate $\Phi(x_+) = (\mathcal{K}_{\Delta_t}^u\Phi)(x)$ by
\begin{equation}\label{eq:bilinear-surrogate-discrete}
    \widehat{\Phi}(x_+) \approx A \widehat{\Phi}(x) + B_0 u + \sum_{i=1}^m u_i (B_i - A) \widehat{\Phi}(x)
\end{equation}
with $B_0=\begin{bmatrix}
    b_1 & \cdots & b_m
\end{bmatrix}$. 
Besides the errors due to finite data and finitely many observables, the representation~\eqref{eq:bilinear-surrogate-discrete} introduces an additional error since the Koopman operator is only approximately control affine and, thus, admits the bilinear structure only approximately~\cite{peitz:otto:rowley:2020,bold:grune:schaller:worthmann:2025}.
The occurring error terms are analyzed in the following.

\subsection{Certified error bounds}\label{sec:error-bounds}
In this section, we present rigorous error bounds for the derived data-driven bilinear surrogate models.
To this end, we define  
\begin{equation}
    (\sigma \widehat{\Phi})(x) 
    = \begin{cases}
        \ddt{}\widehat{\Phi}(x) & \text{continuous time} \\
        \widehat{\Phi}(x_+) & \text{discrete time}
    \end{cases}
\end{equation}
for a combined representation of both continuous and discrete time. 
Then, as shown in~\cite{strasser:schaller:worthmann:berberich:allgower:2025,strasser:schaller:worthmann:berberich:allgower:2024b}, the structure in~\eqref{eq:structure-generator} and~\eqref{eq:structure-operator} allows to formulate the bilinear surrogate model
\begin{equation}\label{eq:bilinear-surrogate}
    (\sigma\widehat{\Phi})(x) = A \widehat{\Phi}(x) +  B_0 u + \tB (u\kron \widehat{\Phi}(x)) + r(x,u),
\end{equation}
where $\tB=\begin{bmatrix}
    B_1 - A & \cdots & B_m - A
\end{bmatrix}$ and $\kron$ denotes the Kronecker product, which we introduced for a compact representation.
In contrast to~\eqref{eq:bilinear-surrogate-continuous} and~\eqref{eq:bilinear-surrogate-discrete}, the formulation in~\eqref{eq:bilinear-surrogate} is \emph{exact} as we account for the residual $r(x,u)$ including all occurring errors. 
Note that the residual may (exemplarily for the continuous-time setting) be decomposed via
\begin{align}
    \begin{bmatrix}0\\r(x,u)\end{bmatrix} 
    &= ((\mathcal{L}^u-L^u_d)\Phi)(x)
    \nonumber\\
    &= ((\mathcal{L}^u - P_\mathbb{V}\mathcal{L}^u_{\vert\mathbb{V}} + P_\mathbb{V}\mathcal{L}^u_{\vert\mathbb{V}} - L^u_d)\Phi)(x),
\label{eq:residual-decomposition}
\end{align}
where $P_\mathbb{V}$ is the $L^2(\mathbb{X})$-orthogonal projection onto $\mathbb{V}$ and $L_d^u$ is defined analogous to~\eqref{eq:generator-control-affine} but with the data-driven estimates $L_d^0,L_d^{e_i}$, $i\in[1:m]$. 
In~\eqref{eq:residual-decomposition}, the first term $\mathcal{L}^u - P_\mathbb{V}\mathcal{L}^u_{\vert\mathbb{V}}$ corresponds to the \emph{projection error} depending on the dictionary $\mathbb{V}$ and the second term $P_\mathbb{V}\mathcal{L}^u_{\vert\mathbb{V}} - L^u_d$ is the \emph{learning} or \emph{estimation error} depending on the number of data points $d$. 
While a proportional bound on the learning error may be obtained with some non-restrictive assumptions on the dictionary, any bound on the projection error inherently depends on the dictionary. 
Thus, we assume a bound on the latter in the following and refer to \cite{philipp:schaller:worthmann:peitz:nuske:2024,kohne:philipp:schaller:schiela:worthmann:2025} for projection error bounds for finite element-based and kernel-based dictionaries.

In the following, we present a \emph{proportional} error bound on the residual which can be used for a robust controller design.
\begin{proposition}\label{prop:proportional-error-bound}
    Suppose that the data samples $\cD$ are drawn i.i.d. in a compact set $\bbX$.
    Let $\tilde{c}_x$, $\tilde{c}_u$ be such that the proportional bound on the projection error
       \begin{align*}
        \|(\mathcal{L}^u \Phi)(x) - (P_\bbV \mathcal{L}^u_{\vert\mathbb{V}})\Phi(x) \| 
        \leq \tilde{c}_x \|\widehat{\Phi}(x)\| + \tilde{c}_u \|u\|
    \end{align*}
    for continuous time or
    \begin{align*}
        \|(\cK_{\Delta t}^u \Phi)(x) - (P_\bbV\cK_{\Delta t}^u\vert_{\mathbb{V}})\Phi(x) \| 
        \leq \tilde{c}_x \|\widehat{\Phi}(x)\| + \tilde{c}_u \|u\|
    \end{align*}
    for discrete time holds.
    
    Then, for any probabilistic tolerance $\delta \in (0,1)$ and data length $d$, there exist constants 
    \begin{equation*}
        \bar{c}_x,\bar{c}_u\in
        \begin{cases}
            \cO(\nicefrac{1}{\sqrt{\delta d}}) & \text{continuous time} \\
            \cO(\nicefrac{1}{\sqrt{\delta d}} + \Delta t^2) & \text{discrete time for $\Delta t >0$}
        \end{cases}
    \end{equation*}
    such that the residual error satisfies the proportional bound
    \begin{equation}\label{eq:proportional-error-bound}
        \|r(x,u)\| \leq c_x\|\widehat{\Phi}(x)\| + c_u\|u\|
    \end{equation}
    with $c_x=\bar{c}_x+\tilde{c}_x$ and $c_u=\bar{c}_u + \tilde{c}_u$ for all $x\in\bbX$ with probability $1-\delta$.        
\end{proposition}
\begin{proof}
    This follows from~\cite[Prop.~5]{strasser:schaller:worthmann:berberich:allgower:2025} and~\cite[Thm.~3.1,~Cor.~3.2]{strasser:schaller:worthmann:berberich:allgower:2024b} in continuous and discrete time, respectively.
\end{proof}
The proportional structure of the bound in~\eqref{eq:proportional-error-bound} is crucial for the following controller design. 
In particular, we exploit that the right-hand side vanishes at our desired equilibrium $(x,u)=0$. 
Further, the constants $\bar{c}_x,\bar{c}_u$ can be made arbitrarily small when collecting sufficiently many data points $d$ (with a small enough sampling rate $\Delta t$ in discrete time). 
Thus, the size of the resulting learning error may be controlled through sampling.
Moreover, the constants $\bar{c}_x,\bar{c}_u$ can be explicitly evaluated for a given data length and sampling rate.
On the other hand, deriving rigorous proportional bounds $\tilde{c}_x,\tilde{c}_u$ on the projection error is an interesting and challenging problem for future research, which may be resolved based on novel deterministic bounds on the full approximation error of a nonlinear surrogate using kernel EDMD and a linear lifting with flexible sampling w.r.t.\ the input~$u$~\cite{kohne:philipp:schaller:schiela:worthmann:2025,bold:philipp:schaller:worthmann:2024}.
In particular,~\cite{strasser:schaller:berberich:worthmann:allgower:2025} shows that a nonlinear data-informed dictionary admits a proportional error bound on the full approximation error of the corresponding bilinear surrogate model, validating the assumed form in Proposition~\ref{prop:proportional-error-bound}.
\section{Koopman-based control with stability guarantees}\label{sec:controller-design}
In this section, we leverage the bilinear surrogate model~\eqref{eq:bilinear-surrogate} with the guaranteed residual error bound~\eqref{eq:proportional-error-bound} according to Proposition~\ref{prop:proportional-error-bound} to derive a robust controller with closed-loop stability guarantees for the underlying nonlinear system.
We separately address the continuous-time (Section~\ref{subsec:controller-design-continuous}) and discrete-time (Section~\ref{subsec:controller-design-discrete}) case.
Throughout the paper, we consider the notion of exponential stability.
\begin{definition}
    The controller $\kappa(x)$ exponentially stabilizes the origin of~\eqref{eq:dynamics-nonlinear} in $\cX$ if there exist $\alpha,\beta>0$ such that $\|x(t;\hat{x},\kappa(x))\|\leq \alpha \|\hat{x}\| e^{-\beta t}$ for all $\hat{x}\in\cX$, $t\geq 0$.
\end{definition}

\subsection{Continuous-time controller design}\label{subsec:controller-design-continuous}
In the following theorem, we design a continuous-time feedback controller, which guarantees closed-loop exponential stability for the uncertain bilinear system~\eqref{eq:bilinear-surrogate} and, thereby, for the underlying nonlinear system~\eqref{eq:dynamics-nonlinear}.
\begin{theorem}\label{thm:stability-continuous-time}
    Let the assumptions of Proposition~\ref{prop:proportional-error-bound} hold and let a probabilistic tolerance $\delta \in (0,1)$ be given.
    If there exist $\alpha\in\bbN$, $\beta\in\bbN_0$ with $\alpha \geq \max\{1,\beta\}$, $P=P^\top \succ 0$ of size $N\times N$, $L\in\bbR[z,2\alpha-1]^{m\times N}$, $\tau\in\SOS_+[z,2\beta]$, and $\rho > 0$ such that $Q(z)\in\mathrm{SOS}[z,2\alpha]^{2N+m}$ with $Q(z)$ defined in~\eqref{eq:stability-continuous-time} and $z\in\bbR^N$,
    \begin{figure*}[!t]
        \small
        \begin{equation}\label{eq:stability-continuous-time}
            Q(z) = 
            \begin{bmatrix}
                -\left(AP + B_0 L(z) + \tB(L(z)\kron z)\right) - \left(AP + B_0 L(z) + \tB(L(z)\kron z)\right)^\top - \rho I_N - \tau(z) I_N & -P & -L(z)^\top \\
                -P & \frac{\tau(z)}{2c_x^2} I_N & 0 \\
                -L(z) & 0 & \frac{\tau(z)}{2c_u^2}I_m
            \end{bmatrix}
        \end{equation} 
        \normalsize
        \hrule
        \vspace*{-0.6\baselineskip}
    \end{figure*}
    then the controller $\mu(x) = L(\widehat{\Phi}(x)) P^{-1} \widehat{\Phi}(x)$ exponentially stabilizes the origin of the continuous-time nonlinear system~\eqref{eq:dynamics-nonlinear} for all initial conditions in $\Omega(c^*)$ with probability $1-\delta$, where
    \begin{subequations}\label{eq:stability-Omega-c-star}
        \begin{equation}\label{eq:stability-Omega}
            \Omega(c) = \left\{
                x\in\bbR^n 
                \mid
                \widehat{\Phi}(x)^\top P^{-1} \widehat{\Phi}(x) \leq c
            \right\}
        \end{equation} 
        and 
        \begin{equation}\label{eq:stability-c-star}
            c^* = \argmax_c \left\{
                c\in\bbR_+ 
                \mid 
                \Omega(c) \subseteq \bbX
            \right\}.
        \end{equation}
    \end{subequations}
\end{theorem}
\begin{proof}
    See Appendix~\ref{sec:app:proof-stability}.
\end{proof}
Theorem~\ref{thm:stability-continuous-time} provides a controller design approach with stability guarantees for the uncertain bilinear surrogate model~\eqref{eq:bilinear-surrogate} and, thereby, for the nonlinear system~\eqref{eq:dynamics-nonlinear}.
Computing the state-feedback law~$\mu$ requires solving an SOS program, which can be tackled using standard solvers; compare Section~\ref{sec:numerics}.
Since the error bound~\eqref{eq:proportional-error-bound} is only guaranteed on the sampling region $\bbX$, the deduced closed-loop guarantees are restricted to the set $\Omega(c^*)\subseteq\bbX$, where the latter is the maximal Lyapunov sublevel set in $\bbX$.
In particular, the set $\Omega$ is positive invariant and, thus, for any initial condition in $\Omega$ the state remains in $\Omega$ for all times.
Here, we obtain stability guarantees with probability $1-\delta$ due to the established error bound~\eqref{eq:proportional-error-bound}, which then holds for all $x\in\bbX$.
Theorem~\ref{thm:stability-continuous-time} improves the existing controller design from~\cite{strasser:schaller:worthmann:berberich:allgower:2025}, which relies on a conservative over-approximation of the bilinear dynamics.
Instead, the proposed approach in Theorem~\ref{thm:stability-continuous-time} directly exploits the bilinearity via SOS techniques and, thereby, offers a less conservative and more flexible controller. 
More precisely, the SOS-based design yields a larger guaranteed region of attraction (RoA) $\Omega$ and allows for larger error bounds $c_x$, $c_u$ in~\eqref{eq:proportional-error-bound}.
The latter is particularly beneficial since it implies that the controller design may be feasible for fewer data samples in comparison to~\cite{strasser:schaller:worthmann:berberich:allgower:2025}.

\subsection{Discrete-time controller design}\label{subsec:controller-design-discrete}
Next, we design a sampled-data controller based on the discrete-time bilinear surrogate model~\eqref{eq:bilinear-surrogate}.
In contrast to the continuous-time case, where a polynomial control law suffices to stabilize an uncertain bilinear system, the discrete-time setting requires a different controller parametrization.
In particular,~\cite{strasser:berberich:allgower:2025} develops a rational controller parametrization in the lifted state, which leads to a sampled-data feedback guaranteeing exponential stability of the uncertain bilinear surrogate model and, hence, of the nonlinear system~\eqref{eq:dynamics-nonlinear}.
The following theorem is derived in~\cite{strasser:berberich:allgower:2025}.
\begin{theorem}\label{thm:stability-discrete-time}
    Let the assumptions of Proposition~\ref{prop:proportional-error-bound} hold and let a probabilistic tolerance $\delta \in (0,1)$ be given.
    Assume $f_c(x,u)$ to be continuous and locally Lipschitz in its first argument around the origin.
    If there exist $\alpha\in\bbN$, $\beta\in\bbN_0$ with $\alpha\geq \beta$, $P=P^\top\succ 0$ of size $N\times N$, $L_\mathrm{n}\in\bbR[z,2\alpha-1]^{m\times N}$, $\tau\in\SOS_+[z,2\beta]$, $u_\dd\in\SOS_+[z,2\alpha]$, and $\rho > 0$ such that~\eqref{eq:stability-discrete-time}
    \begin{figure*}[!t]
        \small
        \begin{equation}\label{eq:stability-discrete-time}
            \begin{bmatrix}
                u_\dd(z) P - \tau(z) I_N & 0 & 0 & u_\dd(z) A P + B_0 L_\mathrm{n}(z) + \tB (L_\mathrm{n}(z) \kron z) \\
                \star & \frac{\tau(z)}{2c_x^2} I_N & 0 & u_\dd(z) P \\
                \star & \star & \frac{\tau(z)}{2c_u^2} I_m & L_\mathrm{n}(z) \\
                \star & \star & \star & u_\dd(z) (P - \rho I_N)
            \end{bmatrix}
            \in\SOS[z,2\alpha]^{3N + m}
        \end{equation} 
        \normalsize
        \hrule
        \vspace*{-0.5\baselineskip}
    \end{figure*}
    holds with $z\in\bbR^N$, 
    then the sampled-data controller 
    \begin{equation}\label{eq:sampled_data_controller}
        \mu_\mathrm{s}(x(t)) = \mu(x(k\Delta t))
        ,\;
        t\in[k\Delta t,(k+1)\Delta t)
        ,\; 
        k\geq 0
    \end{equation}
    with $\mu(x) = \frac{1}{u_\dd(\widehat{\Phi}(x))} L_\mathrm{n}(\widehat{\Phi}(x)) P^{-1} \widehat{\Phi}(x)$ exponentially stabilizes the origin of the continuous-time nonlinear system~\eqref{eq:dynamics-nonlinear} for all initial conditions in $\Omega(c^*)$ with probability $1-\delta$, where $\Omega$ and $c^*$ are defined in~\eqref{eq:stability-Omega-c-star}.
\end{theorem}
\begin{proof}
    See~\cite[Cor.~4]{strasser:berberich:allgower:2025}.
\end{proof}
In contrast to the continuous-time design in Theorem~\ref{thm:stability-continuous-time}, the discrete-time approach in Theorem~\ref{thm:stability-discrete-time} does not require measurements of the state derivative, which significantly improves its practicality.
Moreover,~\eqref{eq:sampled_data_controller} is a sampled-data controller, whereas Theorem~\ref{thm:stability-continuous-time} requires the implementation of a continuous-time state-feedback controller.
For a fixed denominator $u_\dd$,~\eqref{eq:stability-discrete-time} is an SOS program which can be tackled via standard solvers; compare Section~\ref{sec:numerics}.
The denominator $u_\dd$ can be chosen by the user a priori or may be optimized for a fixed $P$ and $\rho$.
Theorem~\ref{thm:stability-discrete-time} provides a significant improvement over the results in~\cite{strasser:schaller:worthmann:berberich:allgower:2024b}, which over-approximate the bilinearity in the surrogate model.
In particular, the SOS-based controller design in Theorem~\ref{thm:stability-discrete-time} leads to an increased guaranteed RoA and can tolerate larger error bounds~\eqref{eq:proportional-error-bound}.
In the case of asymptotic stability, rather than exponential stability, both Theorems~\ref{thm:stability-continuous-time},~\ref{thm:stability-discrete-time} can be relaxed by setting $\rho=0$.
\begin{remark}
    The probabilistic tolerance $\delta$ serves as a hyperparameter in Theorem~\ref{thm:stability-continuous-time} and Theorem~\ref{thm:stability-discrete-time} and is only indirectly used in the respective SOS programs.
    In particular, $\delta$ affects the values of $\bar{c}_x$, $\bar{c}_u$ on the learning error of Proposition~\ref{prop:proportional-error-bound}, i.e., larger values of $\delta$ correspond to tighter error bounds, but less certainty that the true learning error is captured.
    Since the designed controller is only robust against residuals contained in the considered proportional error bound, the controller cannot possibly stabilize the true system. 
    On the other hand, smaller values of $\delta$ yield more accurate error representations with possibly larger constants $c_x$, $c_u$, and thus the robust controller design may not be feasible.
\end{remark}
\section{Koopman-based control with performance guarantees}\label{sec:controller-design-performance}
In this section, we extend the controller design of Section~\ref{sec:controller-design} to address performance specifications for the closed-loop nonlinear system. 
We only address the problem in continuous time and leave the discrete-time treatment for future research.
We introduce the performance input $w\in\bbR^p$ and output $y\in\bbR^q$ as
\begin{subequations}\label{eq:dynamics-nonlinear-performance}
    \begin{align}
        \dot{x} &= f_\mathrm{c}(x,u) + B_w w, \\
        y &= C \widehat{\Phi}(x) + D u + \tD (u \kron \widehat{\Phi}(x)) + D_w w.
    \end{align}
\end{subequations}
We emphasize that the output is restricted to depend bilinearly on the lifted state $\widehat{\Phi}(x)$ and the input $u$, but linearly on the disturbance $w$.
Here, $y$ is chosen by the user and serves as the measure on which the effect of $w$ should be minimized, e.g., the system's state $x$, included in $\widehat{\Phi}(x)$.
For the nonlinear dynamics~\eqref{eq:dynamics-nonlinear-performance}, we follow the discussion of Section~\ref{sec:bilinear-surrogate} and employ Koopman theory to derive a higher-dimensional surrogate model. 
More precisely, we recall the action $\ddt{}\Phi(x) = (\cL^u\Phi)(x)$ of the Koopman generator associated with the nominal system, that is, \eqref{eq:dynamics-nonlinear-performance} with $w=0$. 
Then, we decompose the nominal and disturbed part of the dynamics~\eqref{eq:dynamics-nonlinear-performance} as 
\begin{align}\label{eq:lifted-dynamics-performance-Phi}
    \ddt{}\Phi(x) 
    = (\cL^u\Phi)(x) + \grad \Phi(x)^\top B_w w.
\end{align}
For the first summand, using~\eqref{eq:bilinear-surrogate} yields 
\begin{equation}
    (\cL^u\Phi)(x) 
    = \begin{bmatrix}
        0 \\ A \widehat{\Phi}(x) + B_0 u + \tB (u\kron \widehat{\Phi}(x)) + r(x,u)
    \end{bmatrix}.
\end{equation}
The second summand of~\eqref{eq:lifted-dynamics-performance-Phi} satisfies 
\begin{equation}
    \grad \Phi(x)^\top B_w w 
    = \begin{bmatrix}
        0 \\ \grad \widehat{\Phi}(x)^\top B_w w 
    \end{bmatrix}.
\end{equation}
due to the constant observable $\phi_0(x)\equiv 1$.
Hence, the lifted bilinear surrogate model of the disturbed nonlinear system~\eqref{eq:dynamics-nonlinear-performance} reads
\begin{subequations}\label{eq:dynamics-lifted-performance}
    \begin{align}
        \ddt{}\widehat{\Phi}(x) 
        &= A\widehat{\Phi}(x) + B_0 u + \tB(u\kron \widehat{\Phi}(x)) 
        \nonumber\\ &\qquad 
        + r(x,u) + B_w(\widehat{\Phi}(x)) w, \\
        y &= C \widehat{\Phi}(x) + D u + \tD(u\kron \widehat{\Phi}(x)) + D_w w
    \end{align} 
\end{subequations}
with
\begin{equation}\label{eq:performance-Bwz}
    B_w(z) = \grad \widehat{\Phi}\left(\begin{bmatrix}I_n & 0_{n\times N-n}\end{bmatrix}z\right)^\top B_w.
\end{equation}
To rely on SOS techniques, we restrict ourselves to polynomial observable functions $\widehat{\Phi}(x)$ when studying closed-loop performance.
In particular, $B_w(z)$ is a polynomial in $z$ for any polynomial $\widehat{\Phi}(x)$.
Further, we consider the following quadratic performance specification.
\begin{definition}\label{def:quadratic-performance}
    The system~\eqref{eq:dynamics-lifted-performance} admits quadratic performance with supply rate 
    \begin{equation}\label{eq:supply-rate}
        s(w,y) 
        = \begin{bmatrix}
            w \\ y    
        \end{bmatrix}^\top 
        \begin{bmatrix}
            Q_w & S_w \\ S_w^\top & R_w
        \end{bmatrix}
        \begin{bmatrix}
            w \\ y    
        \end{bmatrix},
    \end{equation}
    where $R_w\succ 0$, $Q_w \prec 0$, if there exist $\varepsilon,\delta>0$ such that
    \begin{equation}\label{eq:quadratic-performance}
            \int_{0}^\infty 
            s(w(t),y(t))
            \dd t 
            \leq 
            - \varepsilon
            \int_{0}^\infty 
            \|w(t)\|^2
            \dd t
    \end{equation}
    for all $w\in\cL_2$.
\end{definition}
For instance, $Q_w = -\gamma^2 I$, $S_w = 0$, and $R_w=I$ correspond to an $\cL_2$-gain\footnote{The $\cL_2$-gain of operator $H$ is defined by the smallest possible $\gamma>0$ such that $\|H(u)\|_{\cL_2} \leq \gamma \|u\|_{\cL_2}$ for all square-integrable functions $u$.} bound $\gamma$ on the channel $w\mapsto y$. 
The following theorem establishes a controller design method guaranteeing exponential stability and quadratic performance of the lifted system~\eqref{eq:dynamics-lifted-performance}.

\begin{theorem}\label{thm:performance-continuous-time}
    Let the proportional error bound~\eqref{eq:proportional-error-bound} hold with some $c_x,c_u>0$ for all $x\in\bbR^n$.
    If there exist $\alpha\in\bbN$, $\beta\in\bbN_0$ with $\alpha\geq \max\{1,\beta\}$, $P=P^\top\succ 0$ of size $N\times N$, $L(z)\in\bbR[z,2\alpha-1]^{m\times N}$, $\tau,\lambda\in\SOS_+[z,2\beta]$, and $\rho,\eta>0$ such that~\eqref{eq:performance-continuous-time} with $Q(z)$ in~\eqref{eq:stability-continuous-time}
    \begin{figure*}[!t]
        \scriptsize
        \begin{equation}\label{eq:performance-continuous-time}
            \left[
                \def\arraystretch{1.15}\begin{array}{ccc|cc}
                    \multicolumn{3}{c|}{\multirow{3}{*}{$Q(z)$}}
                    & -\lambda(z) B_w(z) - \left(C P + DL(z) + \tD(L(z)\kron z)\right)^\top S_w^\top
                    & -\left(C P + DL(z) + \tD(L(z)\kron z)\right)^\top
                    \\
                    & &
                    & 0
                    & 0
                    \\
                    & &
                    & 0
                    & 0
                    \\\hline
                    \star 
                    & \star
                    & \star
                    & -\lambda(z) \left(Q_w + S_wD_w + D_w^\top S_w^\top\right) - \eta I_p
                    & -\lambda(z) D_w^\top
                    \\
                    \star
                    & \star 
                    & \star
                    & \star
                    & \lambda(z) R_w^{-1}
                \end{array}
    	   \right]
            \in\SOS[z,2\alpha]^{2N + m + p + q}
        \end{equation}
        \normalsize
        \medskip
        \vspace*{-0.6\baselineskip}
        \hrule
        \vspace*{-0.6\baselineskip}
    \end{figure*}
    holds with $z\in\bbR^N$, then the controller $\mu(x) = L(\widehat{\Phi}(x)) P^{-1} \widehat{\Phi}(x)$ ensures global exponential stability and quadratic performance of the origin of the continuous-time nonlinear system~\eqref{eq:dynamics-nonlinear-performance}.
\end{theorem}
\begin{proof}
    See Appendix~\ref{sec:app:proof-performance}.
\end{proof}
For the special case of bilinear dynamics in~\eqref{eq:dynamics-nonlinear-performance}, Theorem~\ref{thm:performance-continuous-time} establishes a controller design with guaranteed global exponential stability and quadratic performance for the particular choice $\widehat{\Phi}(x)=x$.
Here, the system dynamics can be either known or estimated using the standard dynamic mode decomposition~\cite{schmid:2010} without lifting.
We note that~\cite{strasser:berberich:allgower:2023b} provides an alternative controller design for bilinear systems with performance guarantees, but the therein developed framework considers only nominal bilinear systems, i.e., it cannot cope with uncertainty. 
Further, it employs an over-approximation of the bilinearity, leading to only \emph{local} performance guarantees and significant conservatism in comparison to the SOS-based approach in Theorem~\ref{thm:performance-continuous-time}.

Theorem~\ref{thm:performance-continuous-time} provides global guarantees assuming that the error bound~\eqref{eq:proportional-error-bound} holds.
However, this error bound can only be guaranteed on a compact set $\bbX$ from which the data are sampled (with probability $1-\delta$).
Therefore, when applying Theorem~\ref{thm:performance-continuous-time} to the nonlinear system~\eqref{eq:dynamics-nonlinear-performance}, it needs to be ensured that the closed-loop trajectory remains within $\bbX$.
This will be shown in the following, assuming that the following mild technical assumption holds.
\begin{assumption}\label{ass:lower-bound-supply-rate}
    There exists a function $\alpha\in\cK$ such that $s(w,y) \geq - \alpha(\|w\|^2)$.\footnote{We define the function class $\cK$ to contain all functions $\alpha:[0,\infty)\to[0,\infty)$ which are continuous, strictly increasing, and satisfy $\alpha(0)=0$.}
\end{assumption}
Assumption~\ref{ass:lower-bound-supply-rate} is satisfied for commonly used supply rates, e.g., the supply rate corresponding to an $\cL_2$-gain bound satisfies $s(w,y)\geq -\gamma^2\|w\|^2$. 
To derive guarantees on a compact set, we restrict ourselves to local performance, i.e., to $w\in \cL_2$ with $\|w\|^2\leq \nu$ for some $\nu > 0$. 
This bound on $w$ is used to ensure robust positive invariance of the set $\Omega(c^*)\subseteq\bbX$, where $\Omega$ and $c^*$ are defined in~\eqref{eq:stability-Omega-c-star}.
\begin{corollary}\label{cor:performance-continuous-time-SafEDMD}
    Let the assumptions of Proposition~\ref{prop:proportional-error-bound} and Assumption~\ref{ass:lower-bound-supply-rate} hold and let a probabilistic tolerance $\delta \in (0,1)$ be given.
    If there exist $\alpha\in\bbN$, $\beta\in\bbN_0$ with $\alpha\geq \max\{1,\beta\}$, $P=P^\top\succ 0$ of size $N\times N$, $L(z)\in\bbR[z,2\alpha-1]^{m\times N}$, $\tau,\lambda\in\SOS_+[z,2\beta]$, and $\rho,\eta>0$ such that~\eqref{eq:performance-continuous-time} with $Q(z)$ in~\eqref{eq:stability-continuous-time} holds with $z\in\bbR^N$, then the controller $\mu(x) = L(\widehat{\Phi}(x)) P^{-1} \widehat{\Phi}(x)$ ensures exponential stability and quadratic performance of the origin of the continuous-time nonlinear system~\eqref{eq:dynamics-nonlinear-performance} for all initial conditions in $\Omega(c^*)$ with probability $1-\delta$, where $\Omega$ and $c^*$ are defined in~\eqref{eq:stability-Omega-c-star}.
\end{corollary}
\begin{proof}
    See Appendix~\ref{sec:app:proof-performance-corollary}.
\end{proof}
Corollary~\ref{cor:performance-continuous-time-SafEDMD} shows that the global stability and performance guarantees from Theorem~\ref{thm:performance-continuous-time} translate into local guarantees when the bilinear surrogate model is estimated using the SafEDMD learning architecture.
The guaranteed region of attraction $\Omega(c^*)$ is the largest Lyapunov sublevel set in $\bbX$.
The proof of Corollary~\ref{cor:performance-continuous-time-SafEDMD} shows that $\Omega(c^*)$ is robust positive invariant under the uncertain bilinear surrogate model and the (bounded) performance input $w$, thus guaranteeing that the state of the nonlinear system~\eqref{eq:dynamics-nonlinear-performance} remains in $\bbX$ for all times.

We note that the continuous-time controller design approaches in Theorem~\ref{thm:stability-continuous-time}, Theorem~\ref{thm:performance-continuous-time}, and Corollary~\ref{cor:performance-continuous-time-SafEDMD} use a polynomial control law, while the discrete-time approach in Theorem~\ref{thm:stability-discrete-time} employs a rational controller parametrization since polynomial controllers cannot globally stabilize bilinear discrete-time systems~\cite{vatani:hovd:olaru:2014}.
Considering rational controllers in continuous time may improve the flexibility of the controller parameterization and, thereby, enhance the feasibility of the conditions in Theorem~\ref{thm:stability-discrete-time}.
Such an extension is straightforward by combining the ideas from Theorems~\ref{thm:stability-continuous-time} and~\ref{thm:stability-discrete-time} and, therefore, we omit it for space reasons.
\section{Numerical example}\label{sec:numerics}
In this section, we demonstrate the advantages of the presented framework by comparing our SOS-based controller design methods in both continuous and discrete time with the corresponding LMI-based approaches in~\cite{strasser:schaller:worthmann:berberich:allgower:2025,strasser:schaller:worthmann:berberich:allgower:2024b}, respectively.
Further, we apply the results from Section~\ref{sec:controller-design-performance} on achieving closed-loop performance.
The corresponding semidefinite programs are solved in Matlab using the toolbox YALMIP~\cite{lofberg:2004} with its SOS module~\cite{lofberg:2009} and the solver MOSEK~\cite{mosek:2022}.\footnote{We provide the code for the simulation examples via \href{https://github.com/rstraesser/SafEDMD-SOS}{https://github.com/rstraesser/SafEDMD-SOS}.}

We consider the nonlinear system used, e.g., in~\cite{brunton:brunton:proctor:kutz:2016}, with dynamics
\begin{subequations}\label{eq:numerics-dynamics}
    \begin{align}
        \dot{x}_1(t) &= \rho x_1(t), \\
        \dot{x}_2(t) &= \lambda(x_2(t) - x_1(t)^2) + u(t)
    \end{align}
\end{subequations}
for $\rho,\lambda\in\bbR$, where we use $\rho=-2$ and $\lambda = 1$ in our simulations.
Since the system dynamics are unknown and only data samples $\cD$ are available, we leverage the SafEDMD procedure elaborated in Section~\ref{sec:bilinear-surrogate} to derive a Koopman-based bilinear surrogate model with certified error bounds. 
To this end, we choose the observable function 
\begin{equation}\label{eq:numerics-observables}
    \Phi(x) = \begin{bmatrix}
        1 & x_1 & x_2 & x_2 - \frac{\lambda}{\lambda-2\rho} x_1^2 & x_1x_2
    \end{bmatrix}^\top.
\end{equation}
Then, we obtain a data-driven bilinear surrogate model in continuous and discrete time according to the discussion in Section~\ref{sec:bilinear-surrogate-continuous} based on state-derivative data~\eqref{eq:data-continuous} and the discussion in Section~\ref{sec:bilinear-surrogate-discrete} based on derivative-free data~\eqref{eq:data-discrete} with sampling rate $\Delta t=\SI{0.01}{}$, respectively.
Here, we draw $d=200$ i.i.d.~data points w.r.t.\ the uniform distribution on the set $\bbX=[-1,1]^2$ for each constant control input $u(t)\equiv \bar{u}\in\{0,1\}$.
For the residual $r$ of the SafEDMD-based surrogate model, we consider Proposition~\ref{prop:proportional-error-bound} and assume the proportional error bound~\eqref{eq:proportional-error-bound} with probability $1-\delta$ for $\delta=0.05$ and 
\begin{equation}\label{eq:numerics-constants-error-bound}
    c_x=c_u = \begin{cases}
        0.1 & \text{continuous time}, \\
        0.006 & \text{discrete time}.
    \end{cases}
\end{equation}
We note that the constants in~\eqref{eq:numerics-constants-error-bound} are the maximal values for which the LMI-based controllers~\cite{strasser:schaller:worthmann:berberich:allgower:2025,strasser:schaller:worthmann:berberich:allgower:2024b}, to which we compare our proposed SOS-based controller, are feasible.
Note that the maximal feasible values of $c_x$ and $c_u$ directly dictate requirements on the employed $\delta$, $d$, and $\Delta t$; cf. Proposition~\ref{prop:proportional-error-bound}.
Then, we solve the proposed SOS programs in Theorem~\ref{thm:stability-continuous-time} and Theorem~\ref{thm:stability-discrete-time} for $\alpha=\beta=1$, where for the latter we choose the denominator $u_\dd\in\SOS_+[\widehat{\Phi}(x),2]$ by including all monomials of degree two.
Figure~\ref{fig:numerics-RoAs} shows the resulting guaranteed RoA $\Omega$ of the \emph{unknown} nonlinear system~\eqref{eq:numerics-dynamics} for both continuous-time and discrete-time surrogate models. 
Further, we depict the resulting guaranteed RoA for the LMI-based controller designs in~\cite{strasser:schaller:worthmann:berberich:allgower:2025,strasser:schaller:worthmann:berberich:allgower:2024b}, which over-approximate the bilinear term inside a user-chosen ellipsoidal region that serves as a superset of the RoA. 
In particular, we apply the approaches from~\cite{strasser:schaller:worthmann:berberich:allgower:2025,strasser:schaller:worthmann:berberich:allgower:2024b} for all $\|\widehat{\Phi}(x)\|^2\leq R_z$ and maximize $R_z$ while retaining feasibility of the approaches.
\begin{figure}[t]
    \centering
    \input{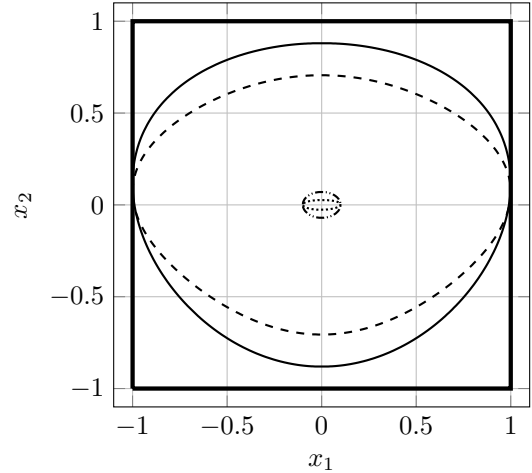}
    \caption{Sampling region $\bbX$~\eqref{plot:compactX}, RoA for the continuous-time surrogate model with SOS controller~\eqref{plot:RoA-continuous-time-SOS} and LMI controller~\eqref{plot:RoA-continuous-time-LMI}, and RoA for the discrete-time surrogate model with SOS controller~\eqref{plot:RoA-discrete-time-SOS} and LMI controller~\eqref{plot:RoA-discrete-time-LMI}.}
    \label{fig:numerics-RoAs}
\end{figure}
Note that the LMI-based approach is only able to guarantee closed-loop stability in a small subset of the sampling region $\bbX$.
In contrast, the proposed SOS-based controllers outperform the LMI-based methods in both continuous and discrete time by achieving a significantly larger guaranteed RoA $\Omega(c^*)\subseteq\bbX$ and by being feasible for larger constants $c_x$, $c_u$ in the proportional error bound~\eqref{eq:proportional-error-bound}.
The latter property is especially advantageous since a larger error bound implies that fewer data samples are required; compare Proposition~\ref{prop:proportional-error-bound}.
Note that the RoA of the closed-loop system in continuous time is larger than in discrete time since the discretization induces an additional approximation step.
On the other hand, the discrete-time approach has the significant practical advantage of not requiring state-derivative measurements and stabilizing the systems via sampled-data feedback.
Note that although the choice of $\Phi$ in~\eqref{eq:numerics-observables} affects the feasibility of controller designs and the shape of the resulting RoAs, the behavior discussed above can also be observed for other observable functions.
\begin{remark}
    The considered nonlinear system in~\eqref{eq:numerics-dynamics}, which was proposed in~\cite{brunton:brunton:proctor:kutz:2016}, has a finite-dimensional linear Koopman representation when choosing the observable function $\Psi(x) = \begin{bmatrix}
        1 & x_1 & x_2 & x_2 - \frac{\lambda}{\lambda-2\rho} x_1^2
    \end{bmatrix}^\top$. 
    In order to obtain a more realistic and practical scenario, we use $\Phi(x)=\begin{bmatrix}
        \Psi(x)^\top & x_1 x_2
    \end{bmatrix}^\top$ in~\eqref{eq:numerics-observables}, which does not admit a finite-dimensional representation. 
    Although we consider this example system for its tutorial value, we emphasize that the proposed methods are applicable more generally.
    For further numerical examples illustrating the proposed SOS-based controller, e.g., with a zone temperature process or an inverted pendulum, we refer to~\cite{strasser:berberich:allgower:2025}.
    Applying the controller design to more complex nonlinear systems is an interesting topic for future research~\cite{peitz:otto:rowley:2020,bevanda:driessen:iacob:toth:sosnowski:hirche:2024}.
\end{remark}

Next, we investigate the achievable performance of the continuous-time controller based on the controller design in Section~\ref{sec:controller-design-performance}.
In particular, we consider the dynamics~\eqref{eq:numerics-dynamics} with the additional performance input $w\in\bbR^2$ and performance output $y=x\in\bbR^2$, yielding the structure in~\eqref{eq:dynamics-nonlinear-performance} for $B_w=I_2$, $C=\begin{bmatrix}I_2 & 0_{2\times 2} \end{bmatrix}$, $D=0$, $\tD=0$, and $D_w=0$.
Note that the respective matrix $B_w(z)$ for the lifted dynamics~\eqref{eq:dynamics-lifted-performance} is obtained via~\eqref{eq:performance-Bwz}. 
The desired performance goal is to minimize the guaranteed $\cL_2$-gain bound $\gamma$, i.e., we choose $Q_w = -\gamma^2 I$, $S_w=0$, and $R_w = I$.
Figure~\ref{fig:numerics-L2-gain} visualizes the guaranteed bound on the $\cL_2$-gain bound with the controller design in Corollary~\ref{cor:performance-continuous-time-SafEDMD} for different constants $c_x$, $c_u$ in the proportional error bound~\eqref{eq:proportional-error-bound} on the residual of the SafEDMD-based surrogate model. 
Note that we restrict ourselves to the choice $c_x=c_u$ for the sake of illustration.
As expected, the minimum $\cL_2$-gain bound which can be guaranteed increases with $c_x$, $c_u$, i.e., increasing uncertainty in the surrogate model deteriorates the guaranteed closed-loop performance.
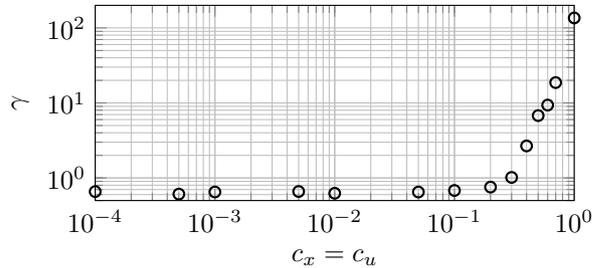
\begin{figure}[t]
    \centering
%
%
\begin{tikzpicture}

\begin{axis}[%
width=0.95\columnwidth,
height=0.5\columnwidth,
xmode=log,
xmin=0.0001,
xmax=1,
xminorticks=true,
xlabel={$c_x=c_u$},
ymode=log,
ymin=0.5,
ymax=200,
yminorticks=true,
ylabel=$\gamma$,
axis background/.style={fill=white},
xmajorgrids,
xminorgrids,
ymajorgrids,
yminorgrids
]
\addplot [black, thick, only marks, mark=o, mark options={solid, black}]
  table[row sep=crcr]{%
1	135.821914672852\\
0.7	18.6718851327896\\
0.6	9.3632698059082\\
0.5	6.77865982055664\\
0.4	2.67246246337891\\
0.3	1.01563453674316\\
0.2	0.755453109741211\\
0.1	0.680665969848632\\
0.05	0.649471282958984\\
0.01	0.625009536743164\\
0.005	0.659799575805664\\
0.001	0.650299072265625\\
0.0005	0.609306335449218\\
0.0001	0.657955169677735\\
};
\end{axis}
\end{tikzpicture}%
    \caption{Guaranteed $\cL_2$-gain bound $\gamma$ of Corollary~\ref{cor:performance-continuous-time-SafEDMD} for different values of $c_x=c_u$ in the proportional error bound~\eqref{eq:proportional-error-bound}.}
    \label{fig:numerics-L2-gain}
\end{figure}
\section{Conclusion}\label{sec:conclusion}
We presented a framework for designing data-driven controllers for unknown nonlinear systems via the Koopman operator.
The proposed approach relied on the SafEDMD framework, which provides a bilinear surrogate model accompanied by rigorous error bounds.
Next, we showed how the resulting uncertain bilinear system can be robustly stabilized, yielding a controller which provably stabilizes the underlying nonlinear system.
We addressed both continuous- and discrete-time setups.
Further, we extended the existing SafEDMD framework by incorporating performance objectives into the controller design.
This allowed us to design controllers based on data which guarantee, for example, a given $\mathcal{L}_2$-gain bound for the controlled nonlinear system.
While we addressed performance only in continuous-time, extending these results to the discrete-time setup is an interesting future research direction.

\begin{funding}
  F.\ Allgöwer is thankful that this work was funded by the Deutsche Forschungsgemeinschaft (DFG, German Research Foundation) under Germany's Excellence Strategy -- EXC 2075 -- 390740016 and within grant AL~316/15-1 -- 468094890.
  K.\ Worthmann gratefully acknowledges funding by the German Research Foundation (DFG; grant WO~2056/14-1 -- 507037103).
  R.\ Strässer thanks the Graduate Academy of the SC SimTech for its support.
\end{funding}

\bibliographystyle{unsrt}
\bibliography{literature}

\appendix 
\section{Proof of Theorem~\ref{thm:stability-continuous-time}}\label{sec:app:proof-stability}
\begin{figure*}[!t]
    \addtocounter{equation}{2}
    \begin{equation}\label{eq:app:proof-stability-1}
        \begin{bmatrix}
            A_K(z)P + PA_K(z)^\top + \rho I_N + \tau(z) I_N & P & PK(z)^\top \\
            P & -\frac{\tau(z)}{2c_x^2} I_N & 0 \\
            K(z)P & 0 & -\frac{\tau(z)}{2c_u^2}I_m
        \end{bmatrix}
        \preceq 0
    \end{equation}
    \vspace*{-0.3\baselineskip}
    \hrule
    \medskip
    \vspace*{-0.3\baselineskip}
    \begin{equation}\label{eq:app:proof-stability-2}
        \begin{bmatrix}
            P^{-1} A_K(z) + A_K(z)^\top P^{-1} + \rho P^{-2} + \tau(z) P^{-2} & I_N & K(z)^\top \\
            I_N & -\frac{\tau(z)}{2c_x^2} I_N & 0 \\
            K(z) & 0 & -\frac{\tau(z)}{2c_u^2}I_m
        \end{bmatrix}
        \preceq 0
    \end{equation}
    \vspace*{-0.3\baselineskip}
    \hrule
    \medskip
    \vspace*{-0.3\baselineskip}
    \addtocounter{equation}{1}
    \begin{equation}\label{eq:proof-continuous-time-Schur-complement-1}
        \begin{bmatrix}
            P^{-1} A_K(z) + A_K(z)^\top P^{-1} + \rho P^{-2} & P^{-1} & I_N & K(z)^\top \\
            P^{-1} & -\frac{1}{\tau(z)}I_N & 0 & 0 \\
            I_N & 0 & -\frac{\tau(z)}{2c_x^2} I_N & 0 \\
            K(z) & 0 & 0 & -\frac{\tau(z)}{2c_u^2}I_m
        \end{bmatrix}
        \preceq 0
    \end{equation}
    \medskip
    \vspace*{-0.6\baselineskip}
    \hrule
    \vspace*{-0.6\baselineskip}
\end{figure*}
We consider the Lyapunov function candidate $V(x) = \widehat{\Phi}(x)^\top P^{-1} \widehat{\Phi}(x)$ for some positive definite $P^{-1}\succ 0$, where a suitable lower and upper bound on $V(x)$ can be constructed using~\eqref{eq:observables-Lipschitz}; see~\cite[Eq.~46]{strasser:schaller:worthmann:berberich:allgower:2025}.
In the following, we establish the Lyapunov inequality 
\addtocounter{equation}{-6}
\begin{multline}\label{eq:proof-continuous-Lyapunov-inequality}
    \dot{V}(x) = (\tfrac{\dd}{\dd t} \widehat{\Phi}(x))^\top P^{-1} \widehat{\Phi}(x) 
    \\
    + \widehat{\Phi}(x)^\top P^{-1} (\tfrac{\dd}{\dd t} \widehat{\Phi}(x)) 
    \leq - \varepsilon\|x\|^2    
\end{multline}
with some $\varepsilon > 0$ for all $x\in\Omega(c^*)$.
To this end, we establish closed-loop stability of the bilinear surrogate with state $\hat{\Phi}(x)$ before exploiting~\eqref{eq:observables-Lipschitz} to ensure that all trajectories of the original state $x$ decay to zero if $\hat{\Phi}(x)$ decays to zero. First, we define $K(z) = L(z)P^{-1}$ and
\begin{equation}\label{eq:proof-A_K}
    A_K(z) = A + B_0 K(z) + \tB(K(z)\kron z).
\end{equation}
Then, by exploiting $L(z)\kron z = (K(z)\kron z) P$, $Q(z)\in\SOS[z,2\alpha]^{2N+m}$ ensures~\eqref{eq:app:proof-stability-1} for all $z\in\bbR^N$.\footnote{The SOS program $Q(z)\in\SOS[z,2\alpha]^{2N+m}$ can only be feasible for $\alpha\geq \max\{1,\beta\}$. This allows us to compensate for the negative impact of the term $-\tau(z)I_N$ in the upper left element.}
Multiplying from the left and the right by $\blkdiag(P^{-1},I,I)$ leads to~\eqref{eq:app:proof-stability-2} which is equivalent to 
\small
\addtocounter{equation}{2}
\begin{multline}
\hspace*{-0.03\columnwidth}
    \begin{bmatrix}
        P^{-1} A_K(z) + A_K(z)^\top P^{-1} + \rho P^{-2} & I_N & K(z)^\top \\
        I_N & -\frac{\tau(z)}{2c_x^2} I_N & 0 \\
        K(z) & 0 & -\frac{\tau(z)}{2c_u^2}I_m
    \end{bmatrix}
    \\
    - \begin{bmatrix}
        P^{-1} \\ 0 \\ 0
    \end{bmatrix}
    \left(-\tau(z)I_N\right)
    \begin{bmatrix}
        P^{-1} \\ 0 \\ 0
    \end{bmatrix}^\top
    \preceq 0
\end{multline} 
\normalsize
for all $z\in\bbR^n$.
Now, we apply the Schur complement and reorder the block rows and columns to obtain equivalently~\eqref{eq:proof-continuous-time-Schur-complement-1}.
Using again the Schur complement, now w.r.t. the last two block rows and columns,~\eqref{eq:proof-continuous-time-Schur-complement-1} reads equivalently
\addtocounter{equation}{1}
\begin{multline}\label{eq:proof-continuous-time-Schur-complement-2}
    \begin{bmatrix}
        P^{-1} A_K(z) + A_K(z)^\top P^{-1} + \rho P^{-2} & P^{-1} \\
        P^{-1} & -\frac{1}{\tau(z)}I_N
    \end{bmatrix}
    \\
    + \left[\star\right]^\top
    \begin{bmatrix}
        \frac{2c_x^2}{\tau(z)} I_N & 0 \\ 0 & \frac{2c_u^2}{\tau(z)}I_m
    \end{bmatrix}
    \begin{bmatrix}
        I & 0 \\
        K(z) & 0
    \end{bmatrix}
    \preceq 0.
\end{multline}
Note that~\eqref{eq:proof-continuous-time-Schur-complement-2} can be decomposed into 
\small
\begin{multline}\label{eq:proof-continuous-time-Schur-complement-3}
    \left[\star\right]^\top 
    \begin{bmatrix}
        \rho P^{-2} & P^{-1} \\
        P^{-1} & 0
    \end{bmatrix}
    \begin{bmatrix}
        I & 0 \\
        A_K(z) & I
    \end{bmatrix}
    \\
    + \frac{1}{\tau(z)} 
    \left[\star\right]^\top
    \begin{bmatrix}
        2c_x^2 I_N & 0 & 0 \\ 0 & 2c_u^2 I_m & 0 \\ 0 & 0 & -I_N
    \end{bmatrix}
    \begin{bmatrix}
        I & 0 \\
        K(z) & 0 \\
        0 & I
    \end{bmatrix}
    \preceq 0.
\end{multline}
\normalsize
Substituting $z$ by $\widehat{\Phi}(x)\in\bbR^N$, multiplying~\eqref{eq:proof-continuous-time-Schur-complement-3} from the left and the right by $\begin{bmatrix} \widehat{\Phi}(x)^\top & r(x,u)^\top \end{bmatrix}$ and its transpose, respectively, and applying the generalized S-procedure~\cite{tan:2006,boyd:vandenberghe:2004} with multiplier $\tau(\widehat{\Phi}(x))^{-1}$, we deduce
\begin{equation}\label{eq:proof-continuous-Lyapunov-S-procedure}
    \left[\star\right]^\top 
    \begin{bmatrix}
        \rho P^{-2} & P^{-1} \\
        P^{-1} & 0
    \end{bmatrix}
    \begin{bmatrix}
        I & 0 \\
        A_K(\widehat{\Phi}(x)) & I
    \end{bmatrix}
    \begin{bmatrix}
        \widehat{\Phi}(x) \\ 
        r(x,u)
    \end{bmatrix}
    \leq 0
\end{equation}
for all $x\in\bbR^n$ if the residual $r$ satisfies
\small
\begin{equation}\label{eq:proof-continuous-residual-S-procedure}
    \left[\star\right]^\top 
    \begin{bmatrix}
        2c_x^2 I_N & 0 & 0 \\ 0 & 2c_u^2 I_m & 0 \\ 0 & 0 & -I_N
    \end{bmatrix}
    \begin{bmatrix}
        I & 0 \\
        K(\widehat{\Phi}(x)) & 0 \\
        0 & I
    \end{bmatrix}
    \begin{bmatrix}
        \widehat{\Phi}(x) \\ r(x,u)
    \end{bmatrix}
    \geq 0.
\end{equation}
\normalsize
Now, recall $\tfrac{\dd}{\dd t}\widehat{\Phi}(x)=A_K(\widehat{\Phi}(x)) \widehat{\Phi}(x) + r(x,u)$ to observe that~\eqref{eq:proof-continuous-Lyapunov-S-procedure} is equivalent to 
\begin{align}
\hspace*{-0.01\columnwidth}
    \dot{V}(x) 
    \leq -\rho \widehat{\Phi}(x)^\top P^{-2} \widehat{\Phi}(x) 
    &\leq -\rho \sigma_\mathrm{min}(P^{-1})^2 \|\widehat{\Phi}(x)\|^2
    \nonumber\\
    &\leq -\rho \|P\|_2^2 L_\Phi^2\|x\|^2,
\end{align}
where we exploit~\eqref{eq:observables-Lipschitz}.
Hence, we have established the desired Lyapunov inequality~\eqref{eq:proof-continuous-Lyapunov-inequality} with $\varepsilon=\rho\|P\|_2^2L_\Phi^2$ for all $r$ satisfying~\eqref{eq:proof-continuous-residual-S-procedure}.
Since all $r$ satisfying the residual bound~\eqref{eq:proportional-error-bound} also satisfy 
\begin{equation}\label{eq:proof-continuous-error-bound-relaxation}
    \|r(x,u)\|^2 \leq 2 c_x^2 \|\widehat{\Phi}(x)\|^2 + 2 c_u^2 \|u\|^2,
\end{equation}
the residual satisfies~\eqref{eq:proof-continuous-residual-S-procedure} for $u=K(\widehat{\Phi}(x))\widehat{\Phi}(x)$.
It remains to show positive invariance of the set $\Omega(c^*)$. 
In particular, the proportional error bound~\eqref{eq:proportional-error-bound} is only guaranteed to hold on $\bbX$ with probability $1-\delta$. 
Hence, guarantees for the nonlinear system can only be deduced if the closed-loop system remains in $\bbX$ for all times.
To this end, we define the Lyapunov sublevel set $\Omega(c)$ in~\eqref{eq:stability-Omega} and maximize the value $c$ according to~\eqref{eq:stability-c-star} such that $\Omega(c^*)$ is the maximal Lyapunov sublevel set in $\bbX$. 
Then, $\Omega(c^*)$ characterizes the largest guaranteed region of attraction in the sampling region $\bbX$, which guarantees closed-loop exponential stability of the nonlinear system~\eqref{eq:dynamics-nonlinear} with probability $1-\delta$.
\qed
\section{Proof of Theorem~\ref{thm:performance-continuous-time}}\label{sec:app:proof-performance}
\begin{figure*}[!t]
    \small
    \addtocounter{equation}{3}
    \begin{equation}\label{eq:app:proof-performance-1}
    \hspace*{-0.01\linewidth}
        \begin{bmatrix}
            P^{-1} A_K(z) + A_K(z)^\top P^{-1} + \rho P^{-2} & \star & \star & \star & \star \\
            B_w(z)^\top P^{-1} + \frac{1}{\lambda(z)}S_w C_K(z) & \frac{1}{\lambda(z)} H_w + \frac{\eta}{\lambda(z)^2}I_p & \star & \star & \star \\
            I & 0 & -\frac{\tau(z)}{2c_x^2} I_N & \star & \star \\
            K(z) & 0 & 0 & -\frac{\tau(z)}{2c_u^2}I_m & \star \\
            C_K(z) & D_w & 0 & 0 & -\lambda(z) R_w^{-1}
        \end{bmatrix}
        - \begin{bmatrix}
            P^{-1} \\ 0 \\ 0 \\ 0 \\ 0
        \end{bmatrix}
        (-\tau(z)I_N) 
        \left[\star\right]^\top
        \preceq 0
    \end{equation}
    \normalsize
    \vspace*{-0.3\baselineskip}
    \hrule
    \medskip
    \vspace*{-0.3\baselineskip}
    \small
    \begin{equation}\label{eq:proof:performance-continuous-time-Schur-1}
        \begin{bmatrix}
            P^{-1} A_K(z) + A_K(z)^\top P^{-1} + \rho P^{-2} & \star & \star & \star & \star & \star \\
            P^{-1} & -\frac{1}{\tau(z)}I_N & \star & \star & \star & \star \\
            B_w(z)^\top P^{-1} + \frac{1}{\lambda(z)}S_w C_K(z) & 0 & \frac{1}{\lambda(z)} H_w + \frac{\eta}{\lambda(z)^2}I_p & \star & \star & \star \\
            I & 0 & 0 & -\frac{\tau(z)}{2c_x^2} I_N & \star & \star \\
            K(z) & 0 & 0 & 0 & -\frac{\tau(z)}{2c_u^2}I_m & \star \\
            C_K(z) & 0 & D_w & 0 & 0 & -\lambda(z) R_w^{-1}
        \end{bmatrix}
        \preceq 0
    \end{equation}
    \normalsize
    \hrule
    \medskip
    \vspace*{-0.5\baselineskip}
    \small
    \begin{equation}\label{eq:proof:performance-continuous-time-Schur-2}
    \hspace*{-0.015\linewidth}
        \begin{bmatrix}
            P^{-1} A_K(z) + A_K(z)^\top P^{-1} + \rho P^{-2} & \star & \star \\
            P^{-1} & -\frac{1}{\tau(z)}I_N & \star \\
            B_w(z)^\top P^{-1} + \frac{1}{\lambda(z)}S_w C_K(z) & 0 & \frac{1}{\lambda(z)} H_w + \frac{\eta}{\lambda(z)^2}
        \end{bmatrix}
        + \left[\star\right]^\top
        \begin{bmatrix}
            \frac{2c_x^2}{\tau(z)} I_N & 0 & 0 \\ 0 & \frac{2c_u^2}{\tau(z)}I_m & 0 \\ 0 & 0 & \frac{1}{\lambda(z)} R_w
        \end{bmatrix}
        \begin{bmatrix}
            I & 0 & 0 \\
            K(z) & 0 & 0 \\
            C_K(z) & 0 & D_w
        \end{bmatrix}
        \preceq 0
    \end{equation}
    \normalsize
    \medskip
    \vspace*{-0.9\baselineskip}
    \hrule
    \vspace*{-0.6\baselineskip}
\end{figure*}
In the following, we establish exponential stability and quadratic performance of the lifted system~\eqref{eq:dynamics-lifted-performance} and thus, if the proportional error bound~\eqref{eq:proportional-error-bound} holds, of the nonlinear system~\eqref{eq:dynamics-nonlinear-performance}.
To this end, we define the Lyapunov candidate function $V(x) = \widehat{\Phi}(x)^\top P^{-1} \widehat{\Phi}(x)$ and recall $A_K(z)$ in~\eqref{eq:proof-A_K}. 
Further, we define $K(z) = L(z) P^{-1}$ and 
\addtocounter{equation}{-6}
\begin{align}
    C_K(z) 
    &= C + DK(z) + \tD(K(z)\kron z),
    \\
    H_w 
    &= Q_w + S_wD_w + D_w^\top S_w^\top.        
\end{align}
Then, multiplying~\eqref{eq:performance-continuous-time} from the left and the right by
\begin{equation}
    \begin{bmatrix}
        P^{-1} & 0 & 0 & 0 & 0 \\
        0 & 0 & 0 & \frac{1}{\lambda(z)} I_N & 0 \\
        0 & I_N & 0 & 0 & 0 \\
        0 & 0 & I_m & 0 & 0 \\
        0 & 0 & 0 & 0 & I_q
    \end{bmatrix}
\end{equation}
and its transpose, respectively, yields~\eqref{eq:app:proof-performance-1} for all $z\in\bbR^N$. 
We apply the Schur complement and reorder the block rows and columns to obtain equivalently~\eqref{eq:proof:performance-continuous-time-Schur-1}.
Using again the Schur complement, now w.r.t. the three last block rows and columns,~\eqref{eq:proof:performance-continuous-time-Schur-1} is equivalent to~\eqref{eq:proof:performance-continuous-time-Schur-2}. 
Note that~\eqref{eq:proof:performance-continuous-time-Schur-2} can be equivalently decomposed into
\addtocounter{equation}{3}
\begin{align}
    &\left[\star\right]^\top 
    \begin{bmatrix}
        \rho P^{-2} & P^{-1} \\
        P^{-1} & 0
    \end{bmatrix}
    \begin{bmatrix}
        I & 0 & 0 \\
        A_K(z) & I & B_w(z)
    \end{bmatrix}
    \nonumber\\
    &+ \frac{1}{\tau(z)} 
    \left[\star\right]^\top
    \begin{bmatrix}
        2c_x^2 I_N & 0 & 0 \\ 0 & 2c_u^2 I_m & 0 \\ 0 & 0 & -I_N
    \end{bmatrix}
    \begin{bmatrix}
        I & 0 & 0 \\
        K(z) & 0 & 0 \\
        0 & I & 0
    \end{bmatrix}
    \nonumber\\
    &+ \frac{1}{\lambda(z)} 
    \left[\star\right]^\top 
    \begin{bmatrix}
        Q_w + \frac{\eta}{\lambda(z)} & S_w \\ S_w^\top & R_w
    \end{bmatrix}
    \begin{bmatrix}
        0 & 0 & I \\
        C_K(z) & 0 & D_w
    \end{bmatrix}
    \preceq 0.
\label{eq:proof:performance-continuous-time-Schur-3}
\end{align}
Recall~\eqref{eq:supply-rate} for the supply rate $s(w,y)$ and~\eqref{eq:dynamics-lifted-performance} for the lifted (closed-loop) dynamics, i.e.,
\begin{align}
\hspace*{-0.031\columnwidth}
    \,\ddt{} \widehat{\Phi}(x) 
    &= A_K(\widehat{\Phi}(x))\widehat{\Phi}(x) + r(x,u) + B_w(\widehat{\Phi}(x))w,\!\!
    \\
    y &= C_K(\widehat{\Phi}(x))\widehat{\Phi}(x) + D_w w.   
\end{align}
Then, substituting $z$ by $\widehat{\Phi}\in\bbR^N$, multiplying~\eqref{eq:proof:performance-continuous-time-Schur-3} from the left and the right by $\begin{bmatrix} \widehat{\Phi}(x)^\top & r(x,u)^\top & w^\top\end{bmatrix}$ and its transpose, respectively, and applying the generalized S-procedure~\cite{tan:2006} with multiplier $\tau(\widehat{\Phi}(x))^{-1}$ yields
\begin{multline}\label{eq:proof:performance-continuous-time-dissipation-inequality}
    \dot{V}(x) 
    \leq 
    - \rho\|P\|_2^2 \|\widehat{\Phi}(x)\|^2 
    \\
    - \frac{\eta}{\lambda(\widehat{\Phi}(x))^2}\|w\|^2 
    - \frac{1}{\lambda(\widehat{\Phi}(x))} s(w,y)
\end{multline}
for all $x\in\bbR^n$ and $r$ satisfying the proportional bound~\eqref{eq:proportional-error-bound} according to the discussion around~\eqref{eq:proof-continuous-error-bound-relaxation}.
Here, we have used~\eqref{eq:observables-Lipschitz} to bound the norm of $\widehat{\Phi}(x)$ in terms of $x$.
Moreover, we define $\lambda^*=\min_{z\in\bbR^N}\lambda(z)$, where we know $\lambda^* > 0$ since $\lambda\in\SOS_+[z,2\beta]$.
Thus, we deduce for $\varepsilon = \nicefrac{\eta}{\lambda^*}$ the inequality
\begin{equation}\label{eq:proof:performance-continuous-time-dissipation-inequality-3}
    \dot{V}(x) \leq - \frac{\varepsilon}{\lambda(z)} \|w\|^2 - \frac{1}{\lambda(z)} s(w,y).
\end{equation}
Hence, by integrating~\eqref{eq:proof:performance-continuous-time-dissipation-inequality-3} from $t=0$ to $\infty$, we establish~\eqref{eq:quadratic-performance} by using $x(0)=0$ and $\lim_{t\to\infty} x(t)=0$, and, thus, quadratic performance according to Definition~\ref{def:quadratic-performance}.
\qed
\section{Proof of Corollary~\ref{cor:performance-continuous-time-SafEDMD}}\label{sec:app:proof-performance-corollary}
Recall that Theorem~\ref{thm:performance-continuous-time} establishes the Lyapunov inequality~\eqref{eq:proof:performance-continuous-time-dissipation-inequality} for all $x\in\bbX$, $w\in\cL_2$ and, thus, for all $x\in\Omega(c^*)\subseteq\bbX$, $\|w\|^2\leq \nu$.
To show robust positive invariance of $\Omega(c^*)$, we need to construct $\nu$ such that $x(t+\chi)\in\Omega(c^*)$ for all $x(t)\in\Omega(c^*)$, $\|w\|^2\leq \nu$, and $\chi>0$. 
The established Lyapunov inequality~\eqref{eq:proof:performance-continuous-time-dissipation-inequality} leads to 
\begin{equation}
    \dot{V}(x) \leq - \rho\|P\|_2^2 \|\widehat{\Phi}(x)\|^2 + \frac{1}{\lambda(\widehat{\Phi}(x))} \alpha(\|w\|^2).
\end{equation} 
Using $V(x)\leq \|\widehat{\Phi}(x)\|^2 \|P\|_2^{-1}$ or, in particular, $\|\widehat{\Phi}(x)\|^2 \geq V(x)\|P\|_2$, yields
\begin{equation}\label{eq:proof:performance-continuous-time-dissipation-inequality-2}
    \dot{V}(x(t)) \leq -\beta V(x(t)) + \frac{1}{\lambda(\widehat{\Phi}(x(t)))} \alpha(\|w(t)\|^2),
\end{equation}
where we define $\beta = \rho\|P\|_2^3$.
Then, we multiply both sides of~\eqref{eq:proof:performance-continuous-time-dissipation-inequality-2} by $e^{\beta t}$ and use $\|w(t)\|^2\leq \nu := \alpha^{-1}(\lambda^* c^*)$ to obtain 
\begin{equation}
    e^{\beta t}\ddt{}V(x(t)) 
    + \beta e^{\beta t} V(x(t)) 
    \leq 
    c^* e^{\beta t}.
\end{equation} 
Note that the left-hand side of the inequality is just the total time derivative of $e^{\beta t} V(x(t))$.
Thus, integrating both sides on the interval $[t,t+\chi]$ for any $\chi > 0$ yields
\begin{multline}
    e^{\beta(t+\chi)} V(x(t+\chi)) 
    - e^{\beta t}V(x(t)) 
    \\ 
    \leq 
    c^* \int_{t}^{t+\chi} e^{\beta s}\dd s
    =
    c^* e^{\beta t}\left[e^{\beta \chi}-1\right].
\end{multline} 
Hence, 
\begin{equation}
    V(x(t+\chi)) 
    \leq e^{-\beta\chi} V(x(t)) 
    + c^* \left[1 - e^{-\beta\chi}\right].
\end{equation}
For all $x(t)\in\Omega(c^*)$, i.e., $V(x(t))\leq c^*$, we obtain
\begin{equation}
    V(x(t+\chi)) 
    \leq e^{-\beta\chi} c^*
    + c^*\left[1 - e^{-\beta\chi}\right]
    = c^*
\end{equation}
for all $\|w(t)\|^2\leq \nu$.
Thus, we have established the existence of a $\nu>0$ such that the RoA $\Omega(c^*)$ is robust positively invariant.

\vskip10pt
\begin{wrapfigure}{l}{25mm} 
    \includegraphics[width=1in,height=1.25in,clip,keepaspectratio]{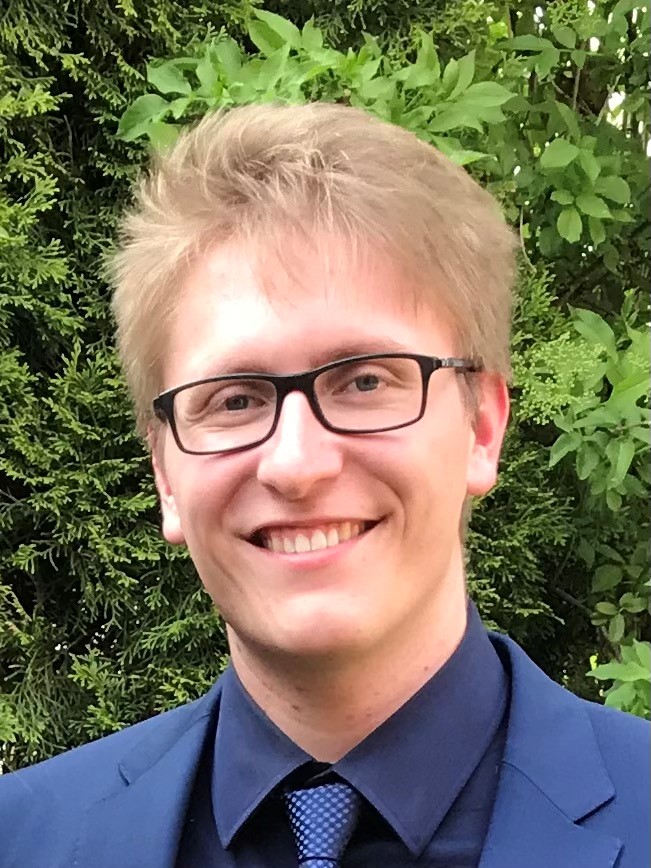}
  \end{wrapfigure}\par
  \textbf{Robin Str\"asser} received a master’s degree in Simulation Technology from the University of Stuttgart, Germany, in 2020. Since 2020, he has been a Research and Teaching Assistant with the Institute for Systems Theory and Automatic Control and a member of the Graduate School Simulation Technology at the University of Stuttgart. His research interests include data-driven system analysis and control, with a focus on nonlinear systems. Robin Strässer received the Best Poster Award at the International Conference on Data-Integrated Simulation Science (SimTech2023).\par
\vskip10pt

\begin{wrapfigure}{l}{25mm} 
    \includegraphics[width=1in,height=1.25in,clip,keepaspectratio]{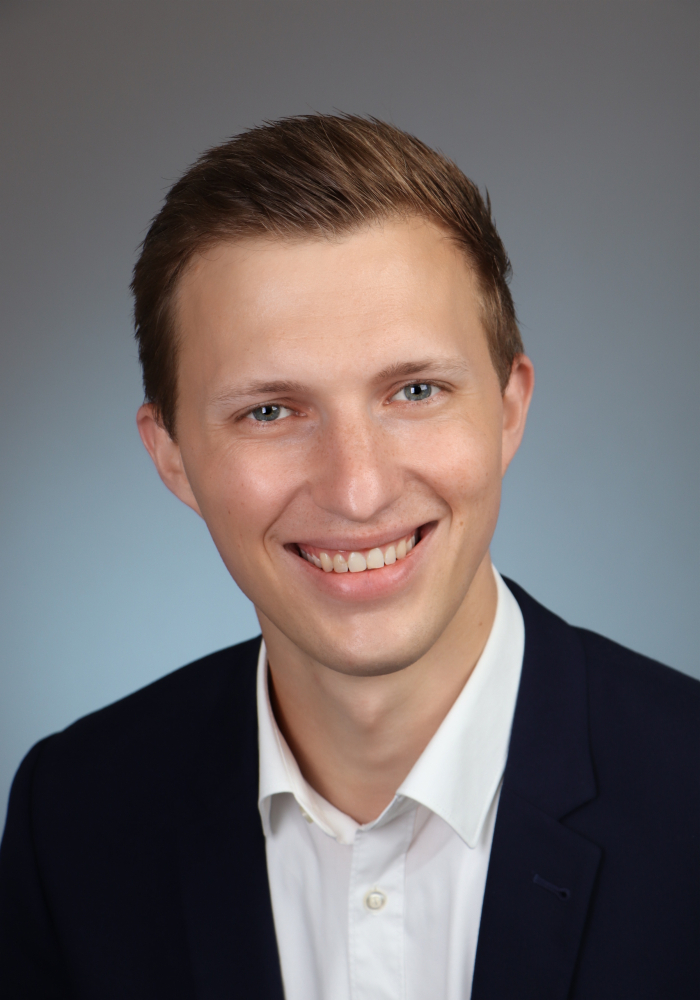}
  \end{wrapfigure}\par
  \textbf{Julian Berberich} is a Lecturer (Akademischer Rat) at the Institute for Systems Theory and Automatic Control at the University of Stuttgart, Germany. He received his Ph.D. in Mechanical Engineering in 2022, and a Master’s degree in Engineering Cybernetics in 2018, both from the University of Stuttgart, Germany. In 2022, he was a visiting researcher at ETH Zürich, Switzerland. He is a recipient of the 2022 George S. Axelby Outstanding Paper Award as well as the Outstanding Student Paper Award at the 59th IEEE Conference on Decision and Control in 2020. His research interests include data-driven analysis and control as well as quantum computing.\par
\vskip10pt

\begin{wrapfigure}{l}{25mm} 
    \includegraphics[width=1in,height=1.25in,clip,keepaspectratio]{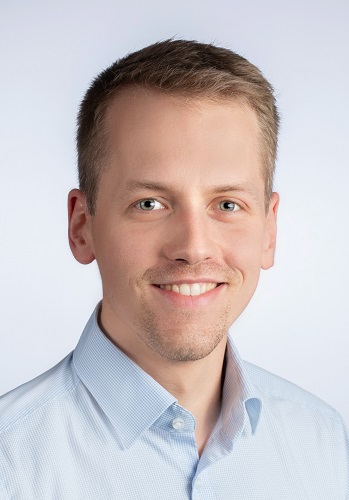}
  \end{wrapfigure}\par
  \textbf{Manuel Schaller} obtained the M.Sc.\ and Ph.D.\ in Applied Mathematics from the University of Bayreuth in 2017 and 2021 respectively. From 2020 to 2023 he held a Lecturer position and a Junior Professorship at Technische Universität Ilmenau, Germany. Since August 2024, he is tenure track assistant professor at Chemnitz University of Technology, Germany. His research focuses on data-driven control with guarantees, port-Hamiltonian systems and efficient numerical methods for optimal control.
  For his research he has been named junior fellow of the GAMM (Society for Applied Mathematics and Mechanics), received the Best Poster Award at the workshop on systems theory and PDEs (WOSTAP 2022) and is an elected member of the Young Academy of the European Mathematical Society (EMS) in 2024-2027.\par
\vskip10pt

\begin{wrapfigure}{l}{25mm} 
    \includegraphics[width=1in,height=1.25in,clip,keepaspectratio]{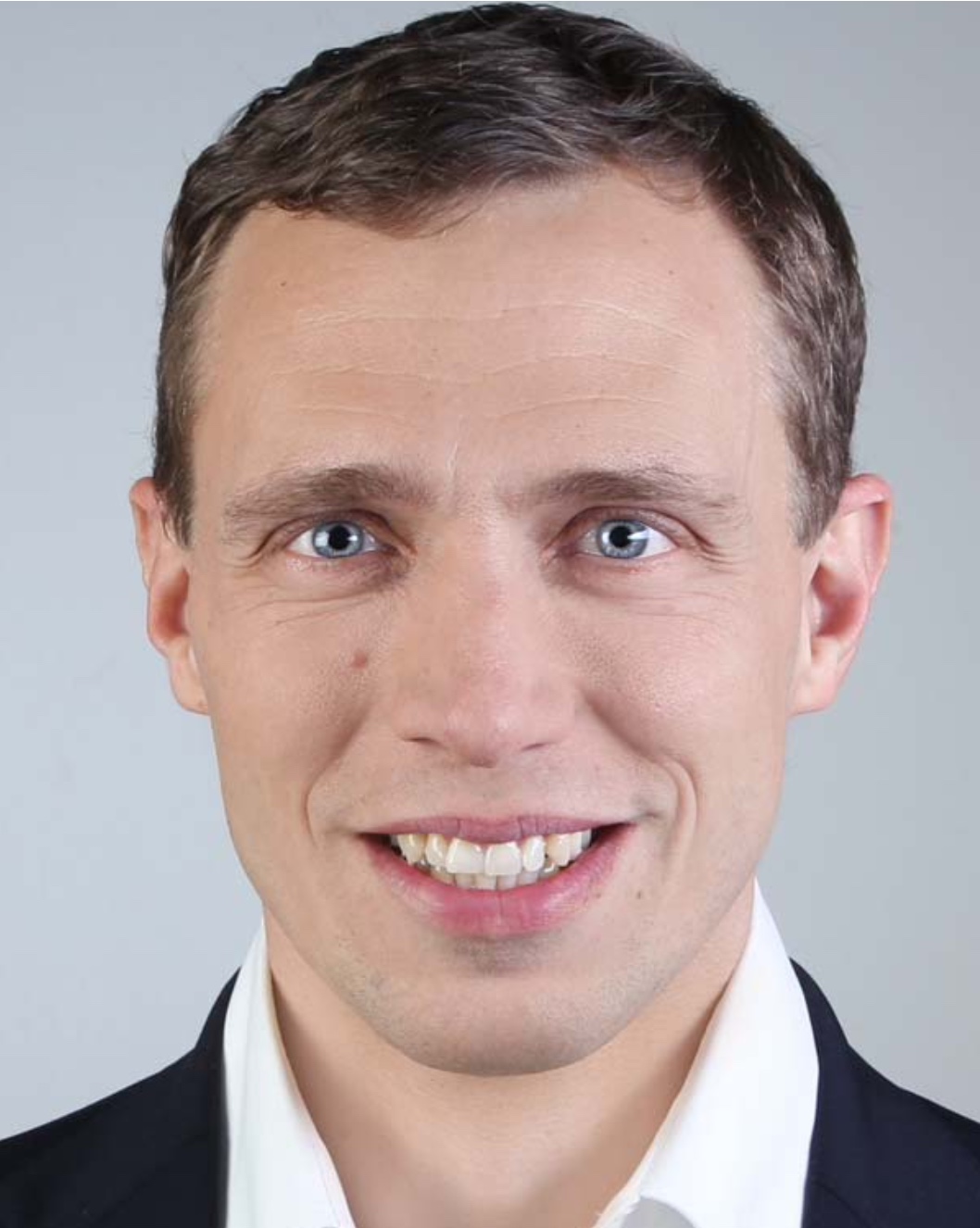}
  \end{wrapfigure}\par
  \textbf{Karl Worthmann} received his Ph.D.\ degree in mathematics from the University of Bayreuth, Germany, in 2012. 2014 he become assistant professor for ''Differential Equations'' at Technische Universität Ilmenau (TU Ilmenau), Germany. 2019 he was promoted to full professor after receiving the Heisenberg-professorship ''Optimization-based Control'' by the German Research Foundation (DFG). He was recipient of the Ph.D.\ Award from the City of Bayreuth, Germany, and stipend of the German National Academic Foundation. 2013 he has been appointed Junior Fellow of the Society of Applied Mathematics and Mechanics (GAMM), where he served as speaker in 2014 and 2015. 
  His current research interests include systems and control theory with a particular focus on nonlinear model predictive control, stability analysis, and data-driven control.\par
\vskip10pt

\begin{wrapfigure}{l}{25mm} 
    \includegraphics[width=1in,height=1.25in,clip,keepaspectratio]{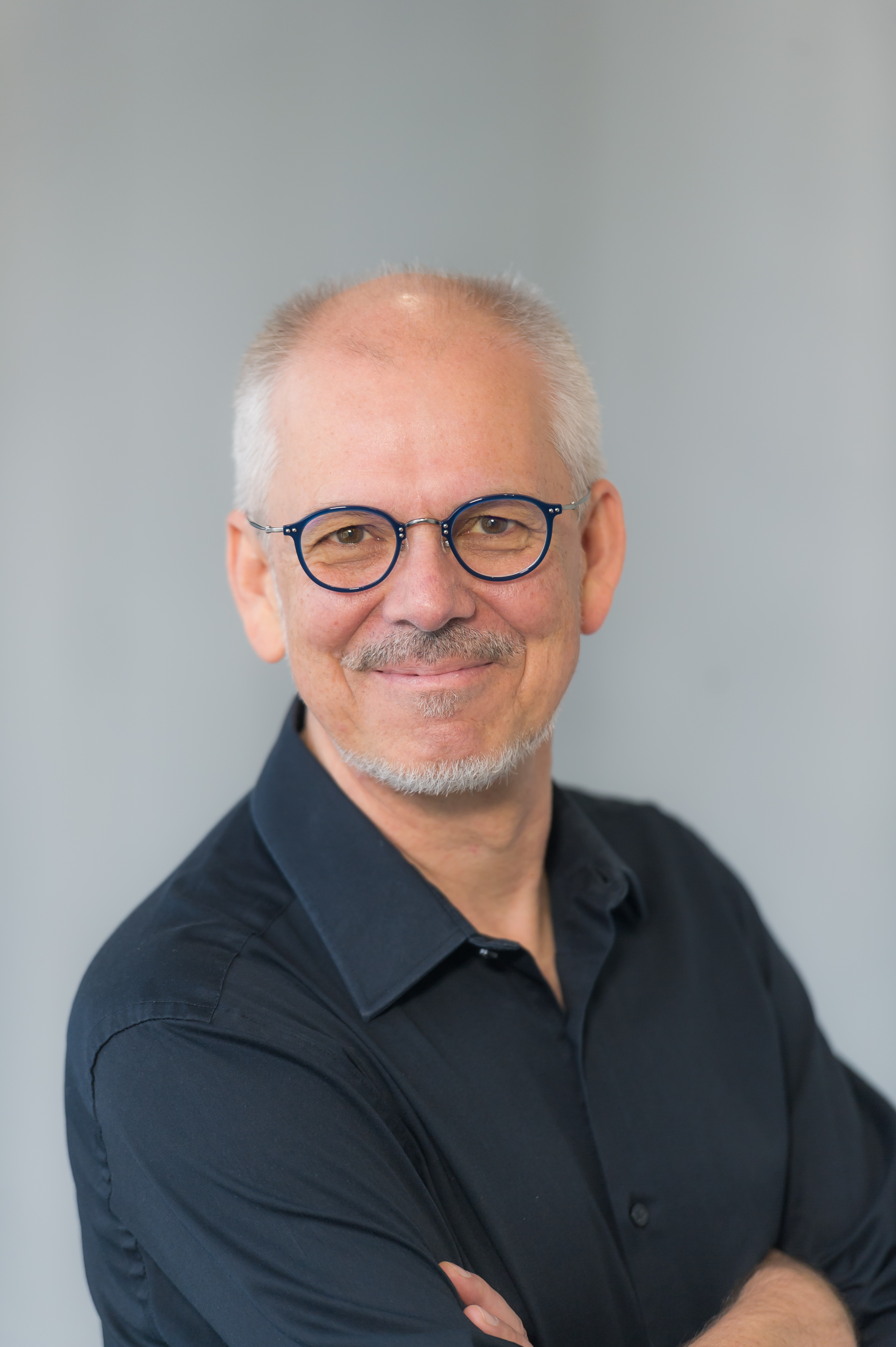}
  \end{wrapfigure}\par
  \textbf{Frank Allg\"ower} studied engineering cybernetics and applied mathematics in Stuttgart and with the University of California, Los Angeles (UCLA), CA, USA, respectively, and received the Ph.D. degree from the University of Stuttgart, Stuttgart, Germany. Since 1999, he has been the Director of the Institute for Systems Theory and Automatic Control and a professor with the University of Stuttgart. His research interests include predictive control, data-based control, networked control, cooperative control, and nonlinear control with application to a wide range of fields including systems biology. Dr. Allgöwer was the President of the International Federation of Automatic Control (IFAC) in 2017–2020 and the Vice President of the German Research Foundation DFG in 2012–2020.\par
\vskip10pt

\end{document}